%% file: DENpaper.tex
\documentclass{article} 
\usepackage{iclr2025_conference,times}
\setcitestyle{numbers,square,sort&compress}

\input{math_commands.tex}

\usepackage{hyperref}
\usepackage{url}
\usepackage{booktabs}
\usepackage{pifont}       
\usepackage{multirow}
\usepackage{geometry}
\usepackage{graphicx}
\usepackage{threeparttable}
\usepackage{makecell}
\usepackage{pdflscape}
\usepackage{amsmath, amssymb}
\usepackage{amsthm}
\usepackage{mathtools}
\newtheorem{proposition}{Proposition}

\newtheorem{remark}{Remark}
\renewcommand{\headrulewidth}{0pt}

\title{
\noindent Dimension expansion for simulation-efficient nanophotonic neural networks}

\author{\hspace*{-0.4in}
Shuo Huang$^{1,2,*}$,
Mahsa Torfeh$^{3,*}$,
Lujia Zhong$^{1,3,*}$,
Michelle L. Povinelli$^{3,\dagger}$,
Owen D. Miller$^{4,\dagger}$,
Chia Wei Hsu$^{3,\ddagger}$ \\
\\
\hspace*{-0.4 in}
$^{1}$Stevens Neuroimaging and Informatics Institute, University of Southern California, Los Angeles, CA 90033, USA \\
\hspace*{-0.4 in}
$^{2}$Department of Biomedical Engineering, University of Southern California, Los Angeles, CA 90089, USA \\
\hspace*{-0.4 in}
$^{3}$Department of Electrical and Computer Engineering, University of Southern California, Los Angeles, CA 90089, USA \\
\hspace*{-0.4 in}
$^{4}$Department of Applied Physics, Yale University, West Haven, CT 06516, USA \\
\\
\hspace*{-0.4 in}
$^{*}$These authors contributed equally to this work. \\
\hspace*{-0.4 in}
$^{\dagger}$Corresponding authors:
\texttt{owen.miller@yale.edu},
\texttt{povinell@usc.edu} \\
\hspace*{-0.4 in}
$^{\ddagger}$Posthumous contribution
}

\iclrfinalcopy 
\begin{document}

\maketitle

\begin{abstract}
Inverse design of nanophotonic structures remains challenging due to the large design space, the nonlinear relationship between structure and optical response, and the computational cost of repeated optimization. Existing deep-learning-based approaches often rely on large pre-generated datasets or collections of optimized structures, limiting their scalability to complex inverse-design problems. Here, we introduce a dimension expansion network (DEN), a fully unsupervised framework for simulation-efficient nanophotonic inverse design. The central idea is to transform low-dimensional design objectives into structured high-dimensional conditioning representations before inverse design is performed, thereby reducing the dimensional mismatch between compact target parameters and high-dimensional nanophotonic structures. The framework is trained directly through differentiable electromagnetic simulations without requiring pre-generated training data. We demonstrate the approach on free-form metalens and asymmetric Y-splitter design tasks. For metalens design, DEN achieves focal intensities comparable to those obtained by adjoint optimization while reducing the number of required simulations by approximately 50\%, and scales from tens to thousands of target focal positions within the same focal region. For Y-splitter design, DEN accurately reproduces arbitrary power-splitting ratios using only 21 training targets and exhibits robust broadband performance. Ablation studies, representation analyses, and theoretical investigations show that dimension expansion improves target sensitivity, increases structural diversity, and reduces mode-collapse-like behavior. These results demonstrate that dimension expansion provides an effective strategy for conditioning inverse-design networks with low-dimensional design objectives and offers a promising approach for photonic inverse design over large continuous target spaces.
\end{abstract}

\section{Introduction}

Nanophotonics uses the interaction of light with subwavelength structures to precisely control electromagnetic waves at the nanoscale~\cite{2015Alu,Alu_Review_2017,Capasso_2015,Review2020Fan,lalanne_1999,2018Vuckovic,revifara,Zhelyeznyakov2021_ScattererField}, enabling a wide range of applications, including sensing~\cite{ma2020_sensors,2022_biosensor}, imaging and displays~\cite{2020_OLED,lee2018_ARVR}, and photonic integrated circuits~\cite{meng2021}. Designing nanophotonic structures with high efficiency and tolerant to fabrication imperfections remains a challenging task, due to the large number of design parameters and the highly nonlinear relationship between structure and optical response. Inverse design methods have been developed to address this challenge, enabling the discovery of high-performance and non-intuitive structures that surpass traditional design approaches~\cite{lalanne_1999,Capasso_2015,Review2020Fan,2018Vuckovic}. 
A variety of optimization-based techniques have been tried, including heuristic methods such as genetic algorithms and particle swarm optimization~\cite{mirjalili2013optical,genetic_2017concurrent,Genetic_jafar2018,lu2018particle,angeris2021heuristic,goudarzi2023inverse}, as well as gradient-based approaches leveraging the adjoint method (including ``topology optimization'')~\cite{Sigmund_2011topology,Yablanovic,ShanhuiFan_2018adjoint,Owen_2020,adjoint_mansouree2021,hammond2022high}. These methods require new iterative optimizations for every variation of a design problem, leading to substantial computational cost, particularly for structures with many degrees of freedom that are not amenable to parameter search methods. Moreover, they are sensitive to initialization and can be trapped in local optima, requiring multiple optimization runs per individual design to identify satisfactory solutions. Requiring binary materials (as opposed to common ``grayscale'' relaxations) and foundry compatibility can further compound these problems \cite{2018Vuckovic,Sigmund_2011topology,piggott2015inverse,sigmund2013topology,christiansen2021inverse,li2026multimode}.

Recently, deep neural networks (DNNs) have been explored for nanophotonic inverse design~\cite{2021_Muskens,khatib_2021_revw_GAN,ma2021deep,jiang2021deep,review2022,chen2022artificial}. Some approaches employ neural networks as surrogate models for ``forward'' electromagnetic simulations; ultrafast simulations then enable iterative optimization at low computational cost~\cite{Zhang2024_Vortex,mlpdnnatom,Zhelyeznyakov2021_ScattererField,An2022_CouplingCNN,Wang2025_DualLayer,transformeratomforw,Wang2025_KAN,mlpatomforward,mlpatomsixout,Peurifoy2018_NanoparticleANN,Tahersima2019_Splitter}. While effective for inverse design, training these models requires large, optimized datasets, often consisting of tens of thousands or even millions of full-wave simulations. Moreover, they do not address challenges such as local optima or repeated optimization for different design targets. Physics-informed neural networks (PINNs) \cite{raissi2019physics,karniadakis2021physics} and neural operators \cite{lu2019deeponet,li2020fourier,kovachki2023neural} have shown effectiveness in incorporating physical constraints into deep learning. However, their application to large-scale inverse design remains challenging. PINNs require repeated evaluation of differential operators over densely sampled domains, which leads to heavy differentiation calculation. Neural operators require extensive training over high-dimensional function spaces, which can be computationally expensive and difficult to be extended to complex tasks.  

An alternative is to use DNNs to directly perform inverse design, training networks to predict optimal structures directly from a desired target or objective, bypassing iterative optimization altogether. This approach can be challenged by the non-uniqueness of inverse problems, where multiple structures can produce similar optical responses, affecting the convergence of the training process~\cite{review2022,mlpgrating}. Several strategies have been proposed to avoid this issue, including filtering training datasets~\cite{kabir2008neural}, dimensionality reduction~\cite{Adibi_2020_npjC}, and restricting the design space~\cite{Gu2021_Bifocal,mlpatominver,dnnatom,Yan2025_MetasurfaceViT,Ou2025_FocusEngineering,cnnatom,mlpatomplas}. These approaches are typically restricted to low-dimensional problems, where they can be quite effective, but do not scale well to complex structures with many degrees of freedom. Tandem neural networks~\cite{Qiu2024_GSST,Ma2018_Chiral,thinmlp,Xu2021_TandemImproved} partially address the non-uniqueness issue by coupling inverse and forward models, but they still require costly training of forward surrogates and have difficulty capturing the complex relationships in high-dimensional design spaces. Generative models, such as generative adversarial networks (GANs)~\cite{goodfellow2014generat} and variational autoencoders (VAEs)~\cite{VAE_2013}, can generate diverse structures from random noise and can explore designs with larger parameter spaces~\cite{review2022,Xiong2024_Multiwavelength,Ma2019_ProbabilisticVAE,cdcganatom,An2021_MultifunctionalGAN,Mall2020_CyclicalDL,Xia2025_DeepLearningHybrid,So2019_cDCGAN,Jiang2019_MetagratingGAN,Wen2020_RobustGAN,ganpatter,Tang2020_IntegratedGAN,ganysplit}. Many such approaches rely heavily on pre-generated training datasets. In practice, constructing such datasets often requires numerous, computationally costly full-wave simulations~\cite{Jiang2019_MetagratingGAN,Wen2020_RobustGAN,ganpatter,Tang2020_IntegratedGAN,ganysplit}. Consequently, the overall computational cost can exceed that of traditional optimization methods. Furthermore, the generated designs tend to have a high degree of similarity to patterns in the training data and may fail to discover new structures with better performances than the structures in the training data.

Simulator-based generative optimization methods, such as GLOnets~\cite{jiang2020simulator}, avoid the need for pre-generated datasets by training directly through electromagnetic simulations. However, these methods are typically driven by random latent inputs and are designed to optimize a single target or a small number of discrete conditions at a time. The stochastic nature of these approaches produces a distribution of candidate structures with a wide range of performance metrics, often dominated by poor- to mediocre-performance designs that require simulation time but do not necessarily help identify the outlier structures desired for a design problem. Typically, post-selection among hundreds of generated structures for the same target is required to identify high-performance designs~\cite{jiang2020simulator}. Although simulator-based generative optimization can reduce the need for large datasets, it remains fundamentally a stochastic search process over a distribution of candidate structures.
In contrast, our objective is not to generate a distribution of candidate structures for a fixed target, but to learn a deterministic conditional mapping from continuously varying target specifications to high-performance designs. Rather than sampling many candidate structures and performing post-selection, the proposed framework directly predicts a design conditioned on the desired target response.
To summarize, existing DNN-based inverse-design methods either have limited ability to efficiently explore high-dimensional design spaces, rely on computationally costly datasets, or operate as stochastic generators that require repeated sampling and post-selection for each target specification. 

In this work, we introduce a fully unsupervised neural network approach with the computational parsimony of adjoint-based iterative optimization approaches, but which can simultaneously learn how to produce optimal structures across an entire class of desired designable devices. The key to our network is a technique we refer to as ``dimension expansion (DEN),'' depicted in Fig~\ref{Figure_propose_method}(a), in which scalar target objectives are lifted to a higher-dimensional space, using specific sinusoidal Fourier ``basis functions'' \cite{vaswani2017attention,tancik2020fourier}, before being fed to attention-based convolutional neural networks. The networks produce \emph{many} designs, one for each scalar target, and back-propagation through the network can improve these designs in parallel, essentially as rapidly as iterative optimization improves individual designs. In a prototypical metalens design application, even with only 1500 \emph{simulations}, this approach can find 25 binarized high-performance designs with variable, independent (sufficiently well-separated) focal points over a $1.6 \lambda \times 1.6 \lambda$ target region. With full-space sampling when all possible focal points are included during training, the network further enables efficient generation of high-performing designs across the entire target space of focal points (6561 designs in $5.33 \lambda \times 5.33 \lambda$ region using only 8300 simulations). Through systematic tests, we show the importance of incorporating dimension expansion, without which the performance of the trained networks collapses, and the field intensity at the target focal point is reduced by about 60\%. We demonstrate the effectiveness of the proposed approach on challenging design tasks, including free-form metalenses and asymmetric Y-splitters, achieving high performance and substantial computational savings, while demonstrating generalization to unseen target conditions in the Y-splitter design task. This work addresses an important challenge in nanophotonic inverse design by improving the conditioning of inverse-design networks with low-dimensional design objectives. The proposed DEN framework demonstrates that dimension expansion can enable efficient learning of high-performance designs across large continuous target spaces. 
Our code is publicly available at: \url{https://github.com/huangshuo343/dimension_expansion_network}.

\begin{figure}[bt]
    \centering
    \includegraphics[width=\textwidth]{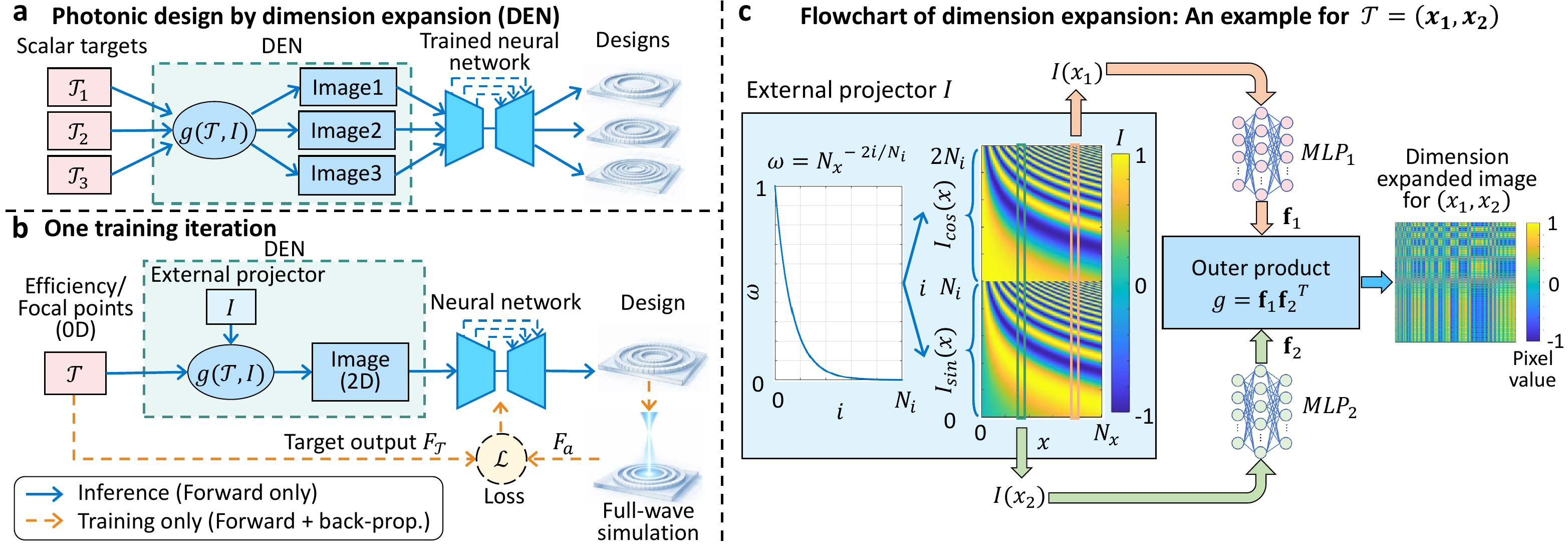}
    \caption{\textbf{The proposed dimension expansion network (DEN).} \textbf{(a)} Structure generation using the trained network under multiple design targets. \textbf{(b)} Training process of the network. The scalar target, such as efficiency or focal-point coordinates, is first embedded into a high-dimensional image and then fed into the neural network. The loss is calculated by directly comparing the actual output of the designed structure, $F_a$, with the target output $F_\mathcal{T}$. \textbf{(c)} Flowchart of dimension expansion, illustrated using the example $\mathcal{T}=(x_1,x_2)$.}
    \label{Figure_propose_method}
\end{figure}

\section{Approach}
In this work, we propose a fully unsupervised deep-learning network for photonic inverse design. As shown in Fig.~\ref{Figure_propose_method}, the proposed DEN contains two parts: 1) Dimension expansion and 2) Inverse design with fully unsupervised neural network. We describe each part in the following sections.


\subsection{Dimension expansion}
\label{sec::dimenexpan}

The goal of our network is to take 0D scalar targets (e.g., focal point coordinates of a metalens or power splitting ratios of Y-splitters) as input and generate the corresponding nanophotonic designs, where the refractive index profile of the design space is the output. 
However, the there is a significant mismatch between the low-dimensional scalar inputs (a few real numbers) and the high-dimensional design-variable output ($10^3$ or more design parameters). 
From a functional perspective, this corresponds to learning a mapping from a very low-dimensional space to a high-dimensional space using a neural network with millions of parameters (often exceeding $10^7$), which is known to be challenging and can lead to poor generalization or training instability~\cite{dimensmismat}.

Increasing the dimensionality of feature representations is a common strategy in neural networks to make complex mappings easier to learn~\cite{reprerevi}. 
For example, convolutional neural networks expand inputs using multiple filters~\cite{convneur,imane}, while multilayer perceptrons and transformer-based models map low-dimensional inputs into higher-dimensional embeddings~\cite{vaswani2017attention,tancik2020fourier}. 
However, these approaches usually assume that the input already has some form of structure, such as spatial grids in images or sequences in tokens. But scalar input targets do not contain such structure. As a result, these standard strategies cannot be directly applied at the input stage. 
Therefore, we need a different way to expand the input dimension, one that can construct structured, high-dimensional representations directly from scalar inputs.
A natural approach is to enrich the input by introducing structured representations that provide more informative features for learning. 
In a preliminary version of this work, we used the electric field obtained by using the target response as the source and propagating it through a predefined initial structure as the input representation~\cite{denwasim}. 
This approach is based on the observation that the electric field provides a spatially structured representation that is informative about the desired structure, which facilitates learning the mapping from targets to designs.
This method becomes less practical as the number of targets increases, since obtaining the electric field requires performing a full-wave simulation for each target. 
In addition, the resulting electric field depends on the choice of the propagation setup, which may reduce its generality as an input representation.

Instead, we transform the scalar targets into a spatially structured representation that can serve as the network input. 
We expand the scalar targets in two steps: (1) mapping each value to a 1D representation using an external projector, and (2) building a 2D representation via an outer product. For the external projector, we use sinusoidal positional encoding \cite{vaswani2017attention}; for a scalar target $x$, the encoding is
\begin{equation}
I(x) = \big[ \sin(\omega_i x), \cos(\omega_i x) \big]_{i=1}^{N_i} \in \mathbb{R}^{2N_i},
\end{equation}
where $N_i$ is the number of ``frequencies'' (not related to the electromagnetic frequencies in the design problem), which themselves are chosen as
\begin{equation}
\omega_i = N_x^{-2i/N_i}, \quad i = 1, \dots, N_i,
\end{equation}
with $N_x$ being the number of discrete samples of the scalar targets (e.g., the number of target focal points for the metalens along one dimension). For targets sampled from a two-dimensional space, such as focal points in the example below, we use $\mathcal{T} = (x_1, x_2)$ to denote its coordinates, with each component encoded independently.

This mapping can be viewed as embedding the scalar inputs into a Fourier feature space, where the target information is distributed across multiple frequency bands using sine and cosine functions. 
Fourier-based encodings have been widely used in neural networks to improve the learning of complex functions from low-dimensional inputs~\cite{tancik2020fourier}.
In this way, the scalar inputs are expanded into a set of multi-frequency basis functions, allowing the network to represent variations at different scales. 
This is particularly suitable for our problem, where the mapping from target parameters to optimal structures is often continuous and locally smooth, meaning that small changes in the targets lead to gradual changes in the structure, while still allowing sharper structural variations when needed. 
While Fourier feature encodings have been widely used in coordinate-based representations and implicit neural models to improve the learning of high-frequency functions from low-dimensional inputs~\cite{tancik2020fourier,mildenhall2021nerf}, and sinusoidal encodings are also commonly used as conditioning mechanisms in generative models, we are not aware of any research using Fourier feature encodings for nanophotonic neural networks taking design specifications as inputs.. 
In this work, we adopt this technique to construct structured input representations for inverse photonic design without relying on training base of electromagnetic simulations.

The encoded representations then pass through two multilayer perceptrons (MLPs),
\begin{equation}
\mathbf{f}_1 = \mathrm{MLP}_1\!\left[I(x_1)\right], \quad
\mathbf{f}_2 = \mathrm{MLP}_2\!\left[I(x_2)\right],
\end{equation}
where each MLP consists of two fully connected layers with Sigmoid Linear Unit (SiLU)~\cite{sigline} activations. These MLPs further process the encoded features and adapt them to the inverse design task by learning task-dependent combinations of the multi-frequency components. 
In this way, the network can emphasize the most relevant frequency components for mapping targets to structures, leading to a representation that better captures the underlying target-structure relationship.

To further expand the dimensionality, we construct a 2D feature map from the 1D feature vectors via the outer product
\begin{equation}
g = \mathbf{f}_1 \mathbf{f}_2^{\top} \in \mathbb{R}^{2N_i \times 2N_i}.
\end{equation}
Each matrix element of $g$ is the unique product of one element from $\mathbf{f}_1$ with one element from $\mathbf{f}_2$, 
capturing interactions between all pairs of encoded features and lifting the two 1D representations into a structured 2D feature map. The 2D image $g$ then feeds as the input into the NN (Fig.~\ref{Figure_propose_method}(c)).

\paragraph{Theoretical motivation.}
The proposed dimension expansion serves two complementary purposes. First, the Fourier encoding maps low-dimensional target coordinates into a multi-scale feature space, where nearby targets become more distinguishable through a frequency-dependent similarity structure. Second, the outer-product construction lifts the coordinate embeddings into a higher-dimensional bilinear representation. As discussed in Appendix~\ref{sec::matheanaly}, this lifting expands the linear span of the conditioning representations and reduces pairwise correlations between different target conditions. Together, these effects improve target separability and help preserve target-dependent information throughout training, thereby reducing mode-collapse-like behavior.



\subsection{Fully unsupervised neural network}
\label{sec_net}

To map the dimension-expanded input representations to physical designs, we employ an attention-based convolutional neural network that generates the refractive index profile $n(x,y)$. 
The network follows an encoder-decoder architecture~\cite{unet}, where convolutional residual blocks~\cite{resn} extract multi-scale spatial features, and attention modules~\cite{vaswani2017attention} capture long-range dependencies. 
In the encoder, the input is progressively processed through down-sampling convolutional blocks to learn hierarchical representations, followed by a bottleneck stage that combines convolution and attention operations. 
The decoder restores spatial resolution through up-sampling, while skip connections fuse encoder features to preserve fine structural details. 
The final output head maps the decoded features to the refractive index profile $n(x,y)$.

To accurately model light-matter interactions during training, we perform full-wave simulations by solving Maxwell's equations for TM polarization (scalar electric field) using a finite-difference frequency-domain (FDFD) solver. 
The design region is discretized using finite differences, while the surrounding homogeneous region is modeled using a Fourier-based propagation method to efficiently capture wave propagation in open space:
\begin{equation}
\left[-\frac{\partial^2}{\partial x^2} - \frac{\partial^2}{\partial y^2} - \frac{\omega^2}{c^2}n^2(x,y)\right]E_z(x,y) = i \omega \mu_0 J_z(x,y).
\label{equation1}
\end{equation}
Here, $\omega$ is the angular frequency, $c$ is the speed of light in vacuum, $\mu_0$ is the vacuum permeability, $n(x,y)$ is the refractive index distribution, $E_z(x,y)$ is the electric field, and $J_z(x,y)$ is an effective source term. 
Further details of the numerical formulation are provided in Sec.~\ref{sec_met}.

We consider two representative examples: asymmetrical Y-splitters and metalenses. 
For the Y-splitter, the device divides incoming light into two output waveguides with target transmission coefficients $T_1$ and $T_2$, while for the metalens, the device focuses light to a target focal point $(f_x, f_y)$. 
The corresponding loss functions are defined as follows. 
For the Y-splitter design, directly minimizing $|T_1 - T_{1\_a}|$ and $|T_2 - T_{2\_a}|$ leads to a highly nonconvex optimization landscape; to improve convergence, we additionally encourage high total transmission:
\begin{equation}
\mathcal{L} = |T_1 - T_{1\_a}| + |T_2 - T_{2\_a}| + \big|1 - (T_{1\_a} + T_{2\_a})\big|.
\end{equation}
For the metalens design, letting $E(x,y)$ denote the simulated electric field, we define
\begin{equation}
\mathcal{L} = -\, |E(f_x, f_y)|^2,
\end{equation}
so that minimizing the loss corresponds to maximizing the focal intensity.

The treatment of targets differs between the two tasks. 
For the Y-splitter, the mapping from target parameters to structures is relatively smooth, allowing the model to be trained on a subset of targets and to generalize to unseen ones within the continuous target space. 
For the metalens design, although the target space is also continuous, the corresponding mapping is significantly more complex and exhibits higher-frequency variations in the target-to-structure mapping. 
As a result, reliable interpolation becomes more challenging, and we include a dense set of target points during training to ensure adequate coverage of the design space.
Once trained, the network generates designs directly from input targets without requiring iterative optimization, learning a shared mapping instead of solving each design problem independently.
For continuous target spaces such as the Y-splitter, this enables interpolation between trained samples and generalization to unseen inputs, while for the metalens, the learned mapping efficiently covers a densely sampled set of targets.

This results in an amortized optimization framework~\cite{VAE_2013,amortinf}, where a single trained model replaces repeated per-target optimization processes. Rather than independently optimizing a separate structure for every target condition, DEN learns a shared target-conditioned mapping across the entire design space during training. Consequently, information about multiple design--target pairs can be learned simultaneously, within a single training process. For example, the proposed framework learns metalens designs covering an entire $5.33\lambda \times 5.33\lambda$ focal region using only 8300 training simulations. For comparison with adjoint-based inverse design, we first only consider 121 independent focal points in this region (corresponding to approximately half-wavelength spacing between focal points in both dimensions), and show high-efficiency designs across all 121 with only 8300 training simulations. Moreover, as shown in Appendix Sec.~\ref{sec::pixeldes}, our framework is nearly as efficient with dense target-space distributions, performing nearly as well for 6561 focal positions in the same region with the same 8300 training simulations. Once training is completed, new designs can be generated directly from their target conditions without requiring additional per-target optimization. These results demonstrate that the proposed approach can efficiently learn and represent large collections of target-dependent designs within a single optimization process.

\section{Results}
In this section, we evaluate the proposed dimension expansion network (DEN) on free-form metalens and asymmetric Y-splitter inverse-design tasks. 
We first demonstrate that DEN can efficiently generate many high-performance metalens designs with different target focal points while significantly reducing the number of required simulations compared with adjoint optimization. 
We then investigate the robustness of the proposed dimension expansion using different network backbones. 
Finally, we evaluate DEN on Y-splitter design tasks with arbitrary power splitting ratios and broadband optical performance.

\begin{figure}[bt]
    \centering
    \includegraphics[width=0.6\textwidth]{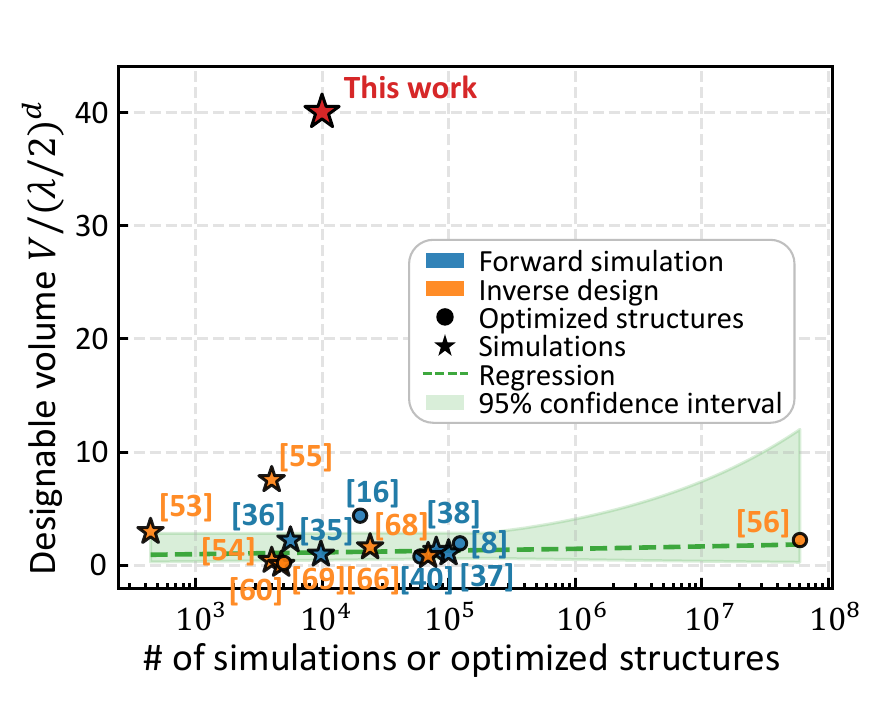}
    \caption{\textbf{Comparison with existing neural-network-based metalens design methods.} We compare existing work and our metalens design outcomes on two axes: (1) the design region's size, $V$, measured in half-wavelength boxes/cubes ($d$ is the simulation dimension): $V / (\lambda/2)^d$; (2) the number of optimized structures or simulations used to generate training data and train the network. ``Forward simulation'' are methods that only use deep learning to predict the forward simulation results of given structures, and ``Inverse design'' are methods that directly output the designed structures. ``Optimized structures'' refer to datasets in which each training structure is produced through optimization, typically requiring many more simulations than the number of training samples, whereas ``simulations'' denote datasets generated by random sampling with only one simulation per data point. The regression line (for correlation of the form $y = a x^{b}$) and the 95\% confidence interval are included.}
    \label{Figure_compare_exist}
\end{figure}

\subsection{Training Details}
\label{sec::tradetai}
We now proceed to train the proposed DEN. All experiments are performed on a single NVIDIA A6000 GPU, with a peak memory consumption of approximately 20~GB during training. The model is optimized end-to-end using differentiable full-wave electromagnetic simulations and does not require any pre-generated training dataset or collection of optimized structures. Instead, the network directly learns from electromagnetic gradients during optimization.

For the metalens design task, we employ a coarse-to-fine target-sampling strategy to improve optimization stability and accelerate convergence. At each training iteration, a target focal position is randomly selected from the target region. During the early stages of training, the focal positions are sampled using a relatively coarse spatial interval. Specifically, for a target focal point region with side length $D$ (shown in Fig.~\ref{Figure_compare_adjoi}(a)), the initial sampling interval is set to $0.5 D$, corresponding to a $3\times3$ focal-point grid within the target region. The interval is subsequently reduced to $0.25D$, corresponding to a $5\times5$ grid, and is further refined during training until the spacing reaches approximately half a wavelength.
This curriculum-like strategy enables the network to first learn coarse target-to-structure relationships before progressively refining the mapping over a denser target distribution. Starting from sparse target sampling also reduces the optimization difficulty during the early stages of training, allowing the network to establish a stable global mapping before resolving finer target-dependent variations. Additional experiments investigating dense target-space coverage are provided in Appendix Sec.~\ref{sec::pixeldes}.

During training, the target coordinates are first processed by the dimension-expansion module and subsequently passed through the neural-network backbone to generate a refractive-index distribution. The generated structure is then transformed into a differentiable binary-like representation using the binarization procedures described in Appendix Sec.~\ref{sec::binar}. The resulting refractive-index profile is embedded into the simulation domain and evaluated using the full-wave solver described in Sec.~\ref{sec_met}. The objective is to maximize the focal intensity at the target position while simultaneously encouraging fabrication-compatible binary structures.
The network parameters are optimized using AdamW with an initial learning rate of $5\times10^{-5}$. Gradients are obtained by back-propagating the electromagnetic loss through the entire simulation and network pipeline. Throughout training, the strength of the binarization constraint is gradually increased, allowing the network to first explore a continuous design space and subsequently converge toward discrete fabrication-compatible structures.

The final DEN checkpoint is selected such that the generated structures achieve binarization degrees comparable to those obtained by direct inverse design using L-BFGS~\cite{lbfgsrevi}. The binarization degree quantifies the extent to which the refractive-index distribution approaches an ideal binary structure, with a value of 1 corresponding to a perfectly binary design and a value of 0 corresponding to a completely intermediate-valued distribution. The formal definition is provided in Appendix Sec.~\ref{sec::binar}. Because the metalens configuration is mirror symmetric with respect to the optical axis, each off-axis focal position has a symmetry-equivalent counterpart. Throughout this work, the two symmetry-related focal positions are evaluated for each method (DEN and inverse design), and the higher focal intensity is reported. This symmetry-aware evaluation protocol is applied consistently to all numerical metalens experiments, including Sec.~\ref{sec_res_met} and Appendix Secs.~\ref{sec::examp} and \ref{sec::pixeldes}.

\begin{figure}[bt]
    \centering
    \includegraphics[width=\textwidth]{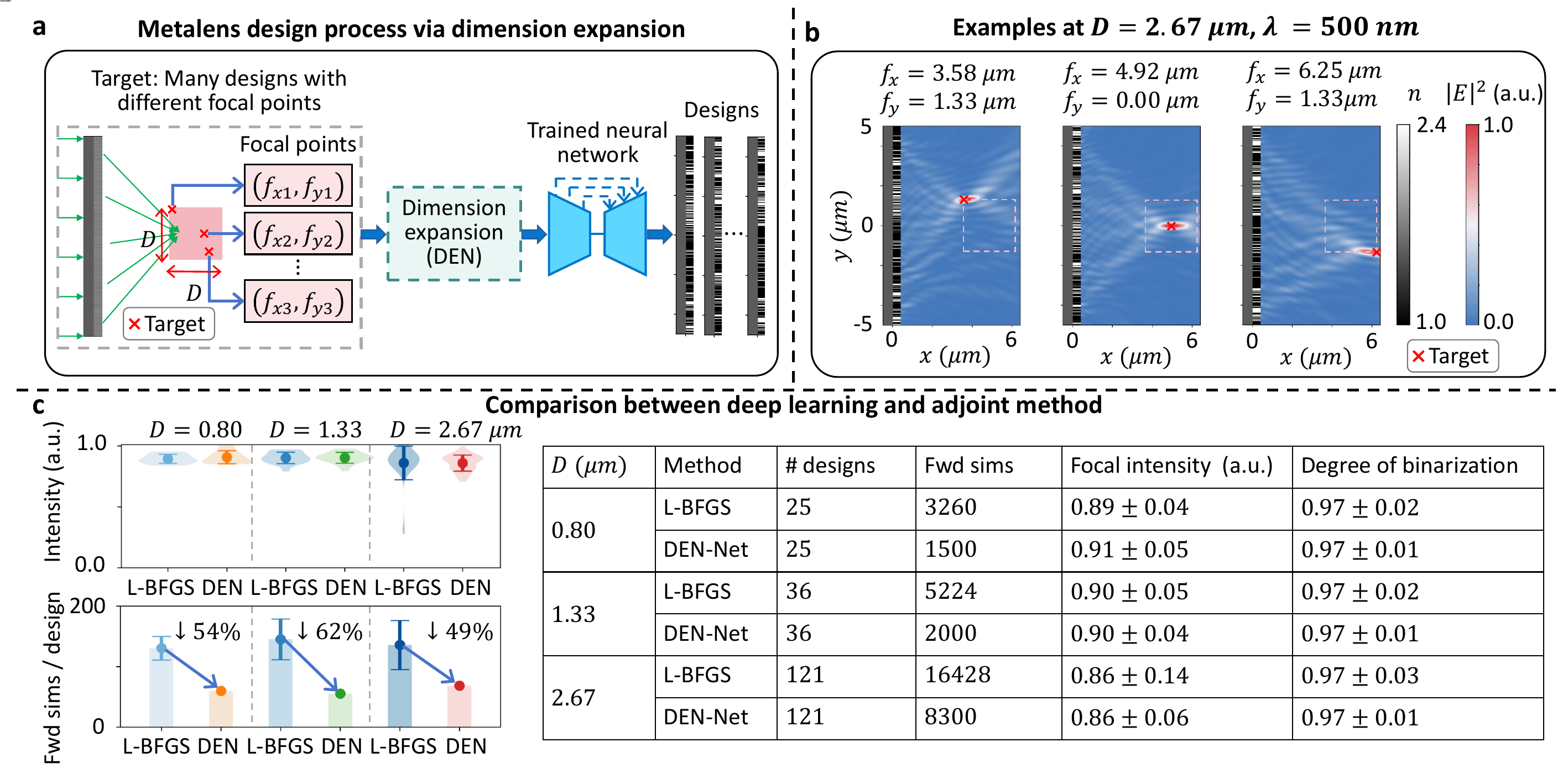}
    \caption{\textbf{Comparison between DEN and adjoint optimization (via L-BFGS) for metalens design.} \textbf{(a)} Schematic illustration of the metalens design task. \textbf{(b)} Three design examples with different target focal-point coordinates and their corresponding output fields. \textbf{(c)} Comparison between DEN and L-BFGS in terms of focal intensity at the full width at half maximum (FWHM) and the number of simulations required to generate the designs. Results are shown for small ($D = 0.80\,\mu\mathrm{m}$), medium ($D = 1.33\,\mu\mathrm{m}$), and large ($D = 2.67\,\mu\mathrm{m}$) focal-point regions.}
    \label{Figure_compare_adjoi}
\end{figure}

\subsection{DEN performance for metalens focusing}
\label{sec_res_met}
We first evaluate the proposed DEN on the inverse design of free-form metalenses. The objective is to generate compact nanophotonic structures that focus incident light to arbitrary target positions within a predefined focal region. As illustrated in Fig.~\ref{Figure_compare_adjoi}(a), the target focal coordinates $(f_x,f_y)$ are transformed through the proposed dimension-expansion module and subsequently mapped to a binary refractive-index distribution by the neural network. The problem is formulated in two dimensions under TM polarization at a wavelength of $\lambda=500~\mathrm{nm}$. The refractive-index distribution is binary with $n\in\{1.0,\,2.4\}$ and is discretized using a spatial resolution of $\Delta x=\lambda/15\approx33.3~\mathrm{nm}$, which also defines the minimum feature size of the design. The metalens design region consists of $12\times301$ pixels, corresponding to a physical size of approximately $0.40~\mu\mathrm{m}\times10.0~\mu\mathrm{m}$ in the propagation ($x$) and transverse ($y$) directions, respectively (Fig.~\ref{Figure_compare_adjoi}(b)). The refractive-index distribution is uniform along the propagation direction, leaving 301 independent design variables along the transverse direction. For the largest target region considered in this work, the focal-point coordinates span a square region with side length $D=2.67~\mu\mathrm{m}$, where the nearest and farthest target positions are located at $f_x=3.58~\mu\mathrm{m}$ and $f_x=6.25~\mu\mathrm{m}$ from the metalens, respectively. These focal distances correspond to effective numerical apertures ranging from approximately $0.81$ to $0.63$. The corresponding off-axis focusing angles span approximately $12.0^\circ$ to $20.4^\circ$ (or $\pm20.4^\circ$ when both positive and negative transverse displacements are considered). Smaller focal regions with side lengths of $D=0.80~\mu\mathrm{m}$ and $D=1.33~\mu\mathrm{m}$ are also considered. All focal regions share the same center, located at $(f_x,f_y)=(4.92~\mu\mathrm{m},\,0.00~\mu\mathrm{m})$. Unlike conventional adjoint optimization, which independently solves a separate optimization problem for every focal point, DEN learns a shared mapping between focal coordinates and optimized structures, enabling many designs to be generated from a single trained model.

A key advantage of DEN is that the optimization cost is amortized across the entire target space. Instead of repeatedly optimizing each focal point independently, the network progressively learns a global target-conditioned mapping between focal coordinates and high-performance structures. Once training is completed, new designs can be generated through a single forward pass without requiring any additional optimization. As a result, the computational cost grows much more slowly than the number of target designs.
To place this efficiency in a broader context, we compare DEN with existing deep-learning-based metalens design approaches in Fig.~\ref{Figure_compare_exist}. The horizontal axis represents the number of simulations or optimized structures required for training, while the vertical axis denotes the normalized designable volume. Existing approaches exhibit a clear empirical scaling trend in which larger design spaces generally require substantially larger simulation datasets or collections of optimized structures. Forward-simulation surrogate models often require tens of thousands to millions of simulations, whereas inverse-design methods typically depend on large datasets of pre-optimized structures.
In contrast, DEN directly performs physics-based inverse design without requiring any pre-generated datasets. As shown in Fig.~\ref{Figure_compare_exist}, the proposed method occupies a distinct region well above the regression trend established by previous approaches. The metalens benchmark considered in this work achieves a normalized designable volume of approximately $40$ using only about $8\times10^3$ simulations, whereas the full-wave volumes designed to date with neural networks methods are smaller than $10$, despite requiring comparable or substantially larger datasets. Additional analyses are provided in the Appendix. Appendix Sec.~\ref{sec::appcomexis} further compares DEN with representative deep-learning-based nanophotonic design methods. Compared with previous approaches, the proposed metalens benchmark achieves one of the largest reported normalized designable volumes ($\approx 40$) while maintaining a comparatively large minimum fabrication feature size ($\approx \lambda/15$) and requiring only about $8\times10^3$ simulations (Fig.~\ref{Figure_compare_exist}, Appendix Sec.~\ref{sec::appcomexis}). The resulting performance lies well above the empirical scaling trend established by existing methods. 

Representative examples are shown in Fig.~\ref{Figure_compare_adjoi}(b) for the largest focal region considered in this work ($D=2.67~\mu\mathrm{m}$). Despite the substantial variation in target focal positions, the generated structures successfully concentrate optical power at the desired locations and produce well-localized focal spots with low background scattering. The focal spots closely coincide with the target coordinates, demonstrating that DEN can reliably generate target-dependent designs across a broad focal region.

We next compare DEN with conventional adjoint optimization using the L-BFGS algorithm. Three focal-region sizes are considered: $D=0.80~\mu\mathrm{m}$, $1.33~\mu\mathrm{m}$, and $2.67~\mu\mathrm{m}$. For each target focal point, L-BFGS independently optimizes a separate structure, whereas DEN jointly learns all target-dependent designs through a single training process.
The quantitative comparison is summarized in Fig.~\ref{Figure_compare_adjoi}(c). Across all focal-region sizes, DEN achieves focal intensities essentially equivalent to those obtained by L-BFGS while requiring substantially fewer full-wave simulations. For the smallest focal region ($D=0.80~\mu\mathrm{m}$), DEN generates 25 designs using only 1500 simulations, corresponding to a $54\%$ reduction relative to the 3260 simulations required by L-BFGS. For $D=1.33~\mu\mathrm{m}$, the simulation count is reduced from 5224 to 2000 ($62\%$ reduction). For the largest focal region ($D=2.67~\mu\mathrm{m}$), DEN generates 121 designs using 8300 simulations, compared with 16428 simulations for L-BFGS, corresponding to a $49\%$ reduction in computational cost.
Importantly, the reduction in simulation cost does not significantly compromise optical performance. DEN achieves average focal intensities of $0.91 \pm 0.05$, $0.90 \pm 0.04$, and $0.86 \pm 0.06$ a.u. for $D=0.80~\mu\mathrm{m}$, $1.33~\mu\mathrm{m}$, and $2.67~\mu\mathrm{m}$, respectively, which are comparable to those obtained by L-BFGS. The generated structures also exhibit high binarization levels of approximately $0.97$, indicating that the differentiable binarization strategy successfully produces fabrication-compatible designs.

Appendix Sec.~\ref{sec::examp} presents a larger collection of metalens designs generated by a single trained network across the focal region, illustrating that DEN can produce diverse target-dependent structures while maintaining consistent focusing performance. Furthermore, Appendix Sec.~\ref{sec::pixeldes} investigates dense target-space sampling, where the same focal region is sampled at substantially higher spatial density. Even when the number of target focal positions increases from $121$ to $6561$, corresponding to dense per-pixel coverage of the same continuous target space, the average focal intensity decreases by only approximately $6\%$. Overall, these results demonstrate that DEN can efficiently learn a target-conditioned mapping between focal coordinates and nanophotonic structures, enabling scalable inverse design across large target spaces while substantially reducing the computational cost relative to conventional optimization-based approaches.

\begin{figure}[bt]
    \centering
    \includegraphics[width=\textwidth]{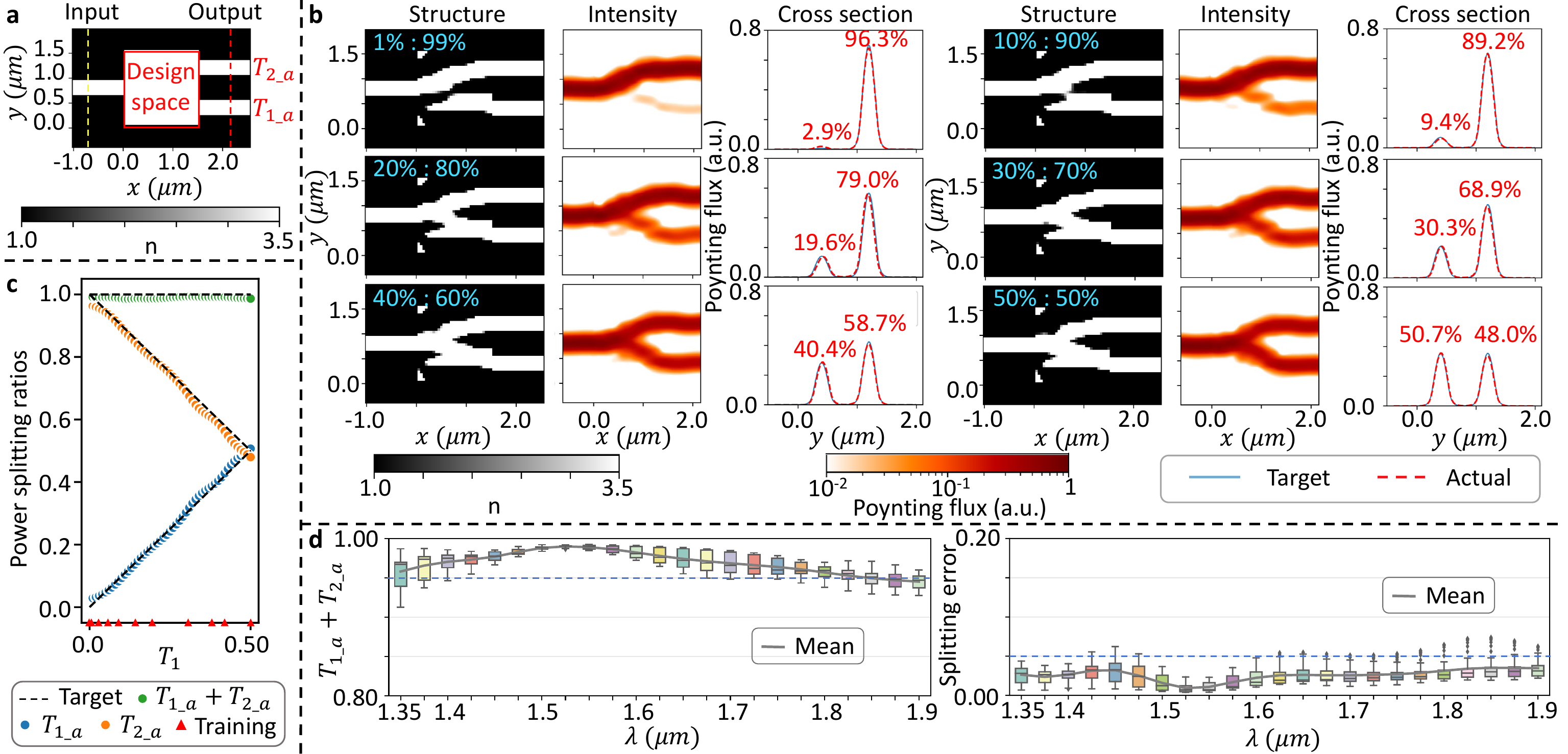}
    \caption{\textbf{Inverse design of Y-splitters using DEN. The model is trained at $\lambda = 1550\,\mathrm{nm}$.} \textbf{(a)} Design configuration. The design region consists of a $30 \times 30 = 900$ pixel square. Light is injected at the yellow line and measured at the blue line. \textbf{(b)} Six representative test examples with different target splitting ratios not used during training. The corresponding Poynting flux distributions and output cross-sectional Poynting flux profiles are also shown. \textbf{(c)} Target and actual splitting ratios, together with total transmission ratios, for test structures with target power-splitting ratios ranging from $1\%:99\%$ to $50\%:50\%$ at $\lambda = 1.55\,\mu\mathrm{m}$. Training targets are marked by red crosses. \textbf{(d)} Broadband performance of structures designed at $\lambda = 1.55\,\mu\mathrm{m}$ over the wavelength range from $\lambda = 1.35\,\mu\mathrm{m}$ to $\lambda = 1.9\,\mu\mathrm{m}$. The splitting error is defined as $(|T_{1\_a}-T_1| + |T_{2\_a}-T_2|)/2$.}
    \label{Figure_yspl_results}
\end{figure}

\subsection{Performances on Y-splitters}
\label{sec::yspli}
We further evaluate the proposed DEN for inverse design of asymmetric Y-splitters. The objective is to generate compact photonic structures that split input optical power into two output waveguides with arbitrary target transmission ratios. As illustrated in Fig.~\ref{Figure_yspl_results}(a), the design region consists of a $30\times30$ binary refractive-index distribution, corresponding to 900 design parameters. Light is injected from the input waveguide on the left side, and the transmitted power is measured at the two output ports on the right side. The network is trained at the target wavelength of $\lambda = 1550~\mathrm{nm}$ using only 21 carefully selected training splitting ratios. The target-selection strategy is described in Appendix Sec.~\ref{sec::tradatyspl}, where the training samples are intentionally concentrated near the extreme splitting conditions to provide improved coverage of the most rapidly varying regions of the design space. During training, the model is optimized end-to-end using full-wave electromagnetic simulations without requiring any pre-generated optimized structures.

Representative examples generated by DEN are shown in Fig.~\ref{Figure_yspl_results}(b). The pictured structures correspond to target splitting ratios ranging from strongly asymmetric ($1\% : 99\%$) to symmetric ($50\% : 50\%$) conditions. Despite the large variation in target power distributions, the generated devices efficiently guide optical power into the desired output ports while maintaining smooth propagation through the splitter region. The corresponding Poynting-flux distributions exhibit limited back reflection and scattering, while the measured output power profiles closely match the prescribed target ratios. These examples demonstrate that the network successfully adapts the device geometry to a wide range of splitting requirements.

A quantitative comparison between the target and actual splitting ratios is presented in Fig.~\ref{Figure_yspl_results}(c). Despite being trained on only 21 target conditions, DEN accurately reproduces splitting ratios across the entire continuous design space. The actual splitting ratios closely overlap the target curve throughout the range from $1\% : 99\%$ to $50\% : 50\%$, including many ratios that are not explicitly included in the training set. At the same time, the total transmission remains close to unity for nearly all generated devices, indicating that the network not only learns the desired power distribution but also preserves high transmission efficiency. These results demonstrate strong interpolation capability between sparsely sampled training targets and suggest that DEN learns a continuous target-conditioned mapping rather than memorizing a discrete set of training examples.

We further evaluate the broadband performance of the generated structures in Fig.~\ref{Figure_yspl_results}(d). Although the network is trained only at $\lambda = 1550~\mathrm{nm}$, the designed Y-splitters maintain high transmission efficiency over a broad wavelength range from $\lambda = 1350~\mathrm{nm}$ to $\lambda = 1900~\mathrm{nm}$. The average total transmission remains above approximately $95\%$ throughout most of the wavelength range and reaches nearly unity around the design wavelength. Meanwhile, the splitting error remains low ($<5\%$) and relatively stable across the spectrum. This behavior indicates that the generated structures are not narrowly optimized for a single operating wavelength but instead exhibit robust broadband performance.

Overall, these results demonstrate that DEN can efficiently generate high-performance Y-splitter designs across a continuous range of splitting ratios using only a small number of training targets and simulations. The network achieves accurate ratio control, high transmission efficiency, strong interpolation capability, and robust broadband performance while avoiding the per-target optimization procedures required by conventional inverse-design approaches.

\begin{figure}[bt]
    \centering
    \includegraphics[width=\textwidth]{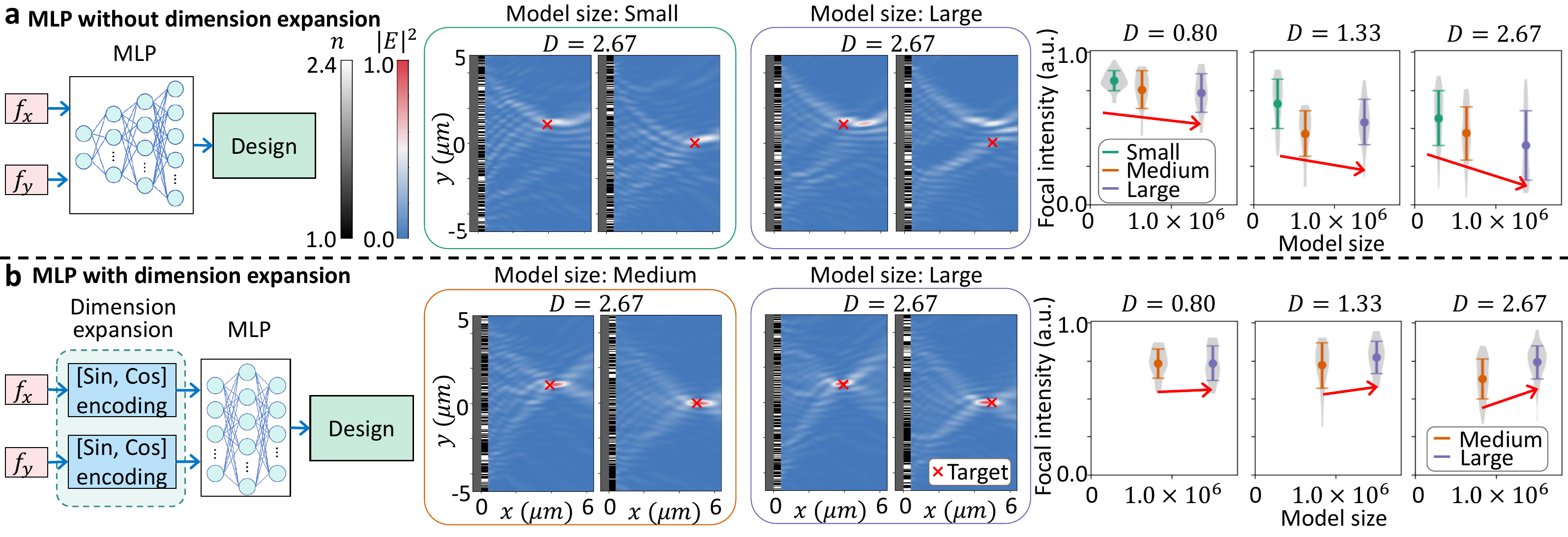}
    \caption{\textbf{Dimension expansion enables effective inverse design even with a simple multilayer perceptron (MLP) backbone.} \textbf{(a)} Performance of MLP-based models without dimension expansion. Performance degrades as the number of model parameters increases or as the focal region becomes larger. \textbf{(b)} Performance of MLP-based models with dimension expansion. The generated structures exhibit improved focusing performance, and the performance further improves as the model size increases.}
    \label{Figure_compare_mlpuseembed}
\end{figure}

\begin{figure}[bt]
    \centering
    \includegraphics[width=\textwidth]{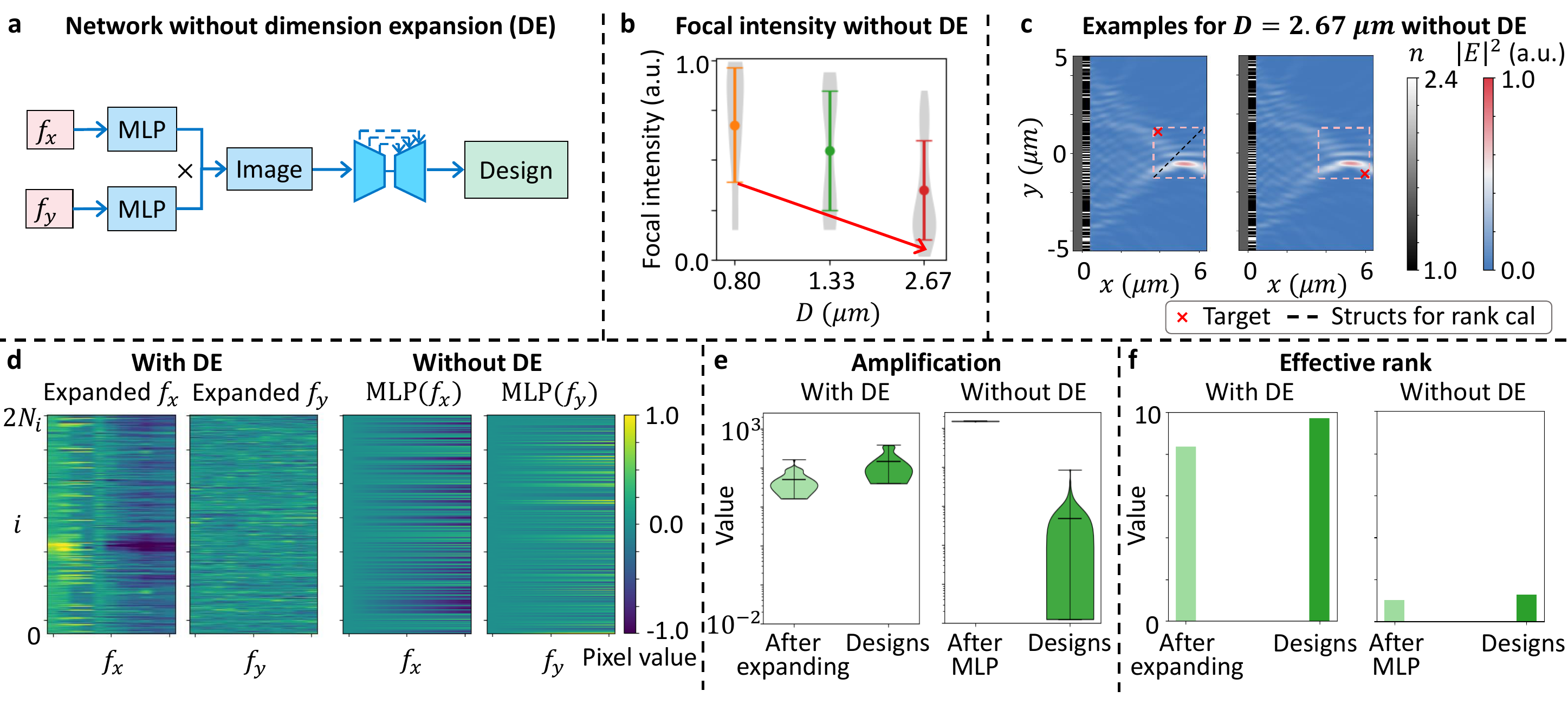}
    \caption{\textbf{Mechanism of dimension expansion (DEN).} \textbf{(a)} Network architecture without dimension expansion. \textbf{(b)} Focal intensity as a function of the target-region size $D$, showing performance degradation and eventual collapse without dimension expansion. \textbf{(c)} Representative failure cases for $D = 2.67\,\mu\mathrm{m}$ without dimension expansion. \textbf{(d)} Comparison of learned conditioning representations. DEN produces more diverse and spatially varying representations than direct MLP mappings. \textbf{(e)} Amplification of target-coordinate variations by DEN, where small differences in target coordinates lead to distinguishable representations and structures. \textbf{(f)} Effective-rank analysis of the generated structures. DEN increases the effective rank, indicating improved structural diversity and reduced mode-collapse behavior.}
    \label{Figure_witho_dimexp}
\end{figure}

\subsection{The role of dimension expansion}
\label{sec_analy_den}


Dimension expansion is the key ingredient enabling the performance gains observed throughout this work. To isolate its contribution, we perform ablation studies using both a simplified multilayer perceptron (MLP) backbone and the full inverse-design network. Figure~\ref{Figure_compare_mlpuseembed}(a) shows a baseline MLP that directly receives the target focal coordinates $(f_x,f_y)$ as input and predicts the corresponding refractive-index distribution. Without dimension expansion, the generated structures produce weak focusing performance, especially as the target focal region becomes larger. For the largest focal region considered here, the resulting electric-field distributions are diffuse and poorly localized, indicating that the MLP struggles to learn an effective mapping from low-dimensional scalar inputs to high-dimensional nanophotonic structures. Increasing the model size does not improve performance and can even lead to further degradation, suggesting that the primary limitation is not network capacity but the low-dimensional target representation itself. In contrast, introducing dimension expansion dramatically improves the performance of the same MLP backbone, as shown in Fig.~\ref{Figure_compare_mlpuseembed}(b). With the expanded input representation, the MLP generates substantially more localized focal spots and cleaner field distributions across all tested target regions. The improvement is observed consistently for small, medium, and large focal regions. Together, these results show that once dimension expansion is introduced, even a simple MLP becomes capable of generating effective target-dependent structures.

The same conclusion is observed in the full inverse-design framework. Figure~\ref{Figure_witho_dimexp}(a) illustrates the baseline model in which the target focal coordinates are fed directly to the network without the Fourier encoding and outer-product expansion introduced in Sec.~\ref{sec::dimenexpan}. All other components of the training procedure are kept unchanged. As shown in Fig.~\ref{Figure_witho_dimexp}(b), removing dimension expansion causes the focal intensity to decrease substantially as the target focal region becomes larger and eventually collapse for $D=2.67~\mu\mathrm{m}$. Representative examples are shown in Fig.~\ref{Figure_witho_dimexp}(c). The generated structures fail to produce well-confined focal spots and instead exhibit diffuse field distributions that are only weakly dependent on the target coordinates. These observations indicate that the network struggles to learn a meaningful target-conditioned mapping when the conditioning variables remain in their original low-dimensional form.

The origin of this behavior can be understood from the learned conditioning representations. As shown in Fig.~\ref{Figure_witho_dimexp}(d), dimension expansion produces richer and more distinctive representations, whereas the baseline model maps many target conditions to substantially more uniform latent representations. Consequently, nearby focal positions become easier to distinguish after dimension expansion. This effect is reflected in the amplification factor shown in Fig.~\ref{Figure_witho_dimexp}(e), where small changes in $(f_x,f_y)$ produce significantly larger variations in the conditioning representations and generated structures. As a result, the generated designs occupy a substantially higher-dimensional subspace, leading to the larger effective ranks observed in Fig.~\ref{Figure_witho_dimexp}(f). Together, these effects increase structural diversity and help prevent the mode-collapse behavior visible in Fig.~\ref{Figure_witho_dimexp}(c).

Additional empirical and theoretical analyses are provided in Appendix Secs.~\ref{sec::addana} and~\ref{sec::matheanaly}. These analyses further support the interpretation that dimension expansion improves target separability and representation diversity, helping explain the substantial performance gains observed throughout this work.

\section{Discussion and future directions}

DEN enables efficient inverse design over large continuous target spaces by transforming low-dimensional design objectives into structured high-dimensional conditioning representations. Across both metalens and Y-splitter benchmarks, the proposed framework achieves high-performance designs while substantially reducing the number of simulations required relative to conventional optimization-based approaches.

Although DEN eliminates the need for pre-generated datasets and substantially reduces the number of simulations required relative to conventional optimization-based approaches, the training process still relies on repeated full-wave electromagnetic simulations. Consequently, the overall computational cost remains coupled to the complexity of the underlying physical solver. This seems like a fundamental physical limit, rather than a limitation of the method, but it is important to acknowledge. As a result, the proposed framework remains constrained by the scale of problems that can be simulated using current electromagnetic solvers~\cite{2018Vuckovic,ShanhuiFan_2018adjoint}.

Another sensible but important point apparent in our two examples is that the complexity of the target-conditioned mapping appears to depend strongly on the underlying inverse-design problem. The Y-splitter task can be learned from only 21 training targets, whereas the metalens problem requires substantially denser target-space coverage to achieve reliable performance across large focal regions. Different inverse-design problems likely possess fundamentally different target-space complexities, even if the corresponding structures contain comparable numbers of design parameters. Understanding the relationship between target-space complexity, sampling density, and model capacity remains an open question.

Third, the current framework learns a deterministic target-to-structure mapping and therefore converges toward a single family of high-performance solutions. While this behavior alleviates the non-uniqueness issue that often complicates inverse-design training, it does not explicitly characterize the full set of feasible structures associated with a given target response. Integrating dimension expansion with probabilistic or generative formulations may provide a means to explore multiple solution branches while retaining the target sensitivity demonstrated in this work~\cite{goodfellow2014generat,VAE_2013}.

More broadly, the present results suggest several opportunities for future investigation. One direction is to examine whether the proposed conditioning strategy remains effective for larger-scale three-dimensional devices. In particular, the current framework could be generalized to three-dimensional nanophotonic structures by encoding three-dimensional focal coordinates $(f_x,f_y,f_z)$ and constructing higher-order tensor representations through multi-way outer products. The resulting conditioning representation would naturally become a three-dimensional tensor rather than a two-dimensional image, which could then be processed using three-dimensional neural-network architectures. Voxelized refractive-index distributions provide a natural extension of the pixelized structures considered in this work. Practical application to such settings will likely require more efficient electromagnetic solvers. Future implementations may benefit from differentiable GPU-accelerated simulation frameworks, hybrid simulation-surrogate approaches, neural operators, or RCWA-based approaches such as TORCWA,~\cite{lu2019deeponet,li2020fourier,kovachki2023neural,torcw} which may help improve the scalability of inverse design for large-scale three-dimensional photonic structures.

A second direction is to extend DEN to substantially larger target spaces involving many simultaneously specified design objectives. For example, focal position, wavelength, polarization response, efficiency, or other optical-performance metrics could be incorporated directly into the conditioning representation. However, a straightforward extension of the current outer-product construction to many target variables would lead to rapidly growing representation dimensionality due to the multiplicative expansion of tensor-product spaces. One possible approach is therefore to generalize the current Fourier-encoding and outer-product framework hierarchically. Rather than constructing a single high-order tensor over all target variables, related objectives could first be grouped according to their physical roles and expanded independently before being fused through learned projection layers. Such hierarchical conditioning representations may provide a scalable way to extend dimension expansion to complex multi-objective inverse-design problems while controlling the growth of representation dimensionality. Because such extensions would likely require substantially denser target-space sampling than considered here, many-input many-output Maxwell solvers such as MESTI~\cite{mes} may offer an efficient way to evaluate large numbers of target conditions during training, potentially making large-scale multi-objective conditioning practical for target-conditioned inverse design.

\section{Methods}
\label{sec_met}

The proposed DEN is trained through direct electromagnetic feedback and therefore requires repeated evaluations of the optical response throughout the optimization process. Consequently, the computational efficiency of the underlying electromagnetic solver plays a critical role in determining the overall training cost. To balance physical accuracy and computational efficiency, we employ a hybrid numerical scheme that combines finite-difference frequency-domain (FDFD) simulations with Fourier-transform-based propagation.

Within the device region, where the refractive index varies spatially and strongly influences the electromagnetic response, Maxwell's equations are solved using an FDFD formulation introduced in Eq.~\ref{equation1} in Sec. ~\ref{sec_net}; the solver accurately captures multiple scattering, interference, and mode-conversion effects inside the nanophotonic structure and provides fully differentiable gradients for back-propagation.
Outside the design region, the medium is homogeneous and the electromagnetic field propagates through free space. In these regions, solving the full-wave problem is unnecessary and computationally inefficient. Instead, we employ a Fourier-transform-based propagation method. The electromagnetic field is first transformed into the spatial-frequency domain, where each spatial-frequency component propagates independently according to the free-space transfer function. The propagated field is then transformed back to the spatial domain using an inverse Fourier transform. This approach efficiently models long-distance propagation while avoiding the computational cost associated with large full-wave simulation domains.

The resulting hybrid formulation combines the strengths of both approaches. The FDFD solver is restricted to the region where complex light--matter interactions occur, while Fourier propagation is used in homogeneous regions where analytical propagation operators are available. This significantly reduces the computational burden of the training process while preserving the physical accuracy required for inverse design. The approach is particularly advantageous for metalens problems, where the focal plane can be located several wavelengths away from the device and would otherwise require a substantially larger simulation domain.
All simulations are implemented within a differentiable framework, allowing gradients to propagate through both the FDFD and Fourier-propagation stages. As a result, the network parameters can be optimized directly using gradient-based methods without requiring surrogate models or numerical approximations of the adjoint sensitivities.

To ensure fabrication feasibility, the refractive-index distribution is progressively driven toward binary material states during training. Direct thresholding would introduce discontinuities into the optimization process and prevent gradient-based learning. Therefore, we employ differentiable binarization strategies that preserve end-to-end differentiability while gradually enforcing fabrication-compatible structures.

The specific binarization formulation depends on the design task. For the metalens problem, we adopt an imaginary-penalty-based approach that suppresses intermediate refractive-index values during optimization. For the Y-splitter problem, we employ a sigmoid-based formulation whose transition becomes progressively sharper throughout training. In both cases, the binarization constraint is introduced gradually, allowing the network to first explore a continuous design space before converging toward discrete material distributions.
These continuation strategies improve optimization stability and enable efficient training of the inverse-design network while maintaining compatibility with binary fabrication processes. The detailed formulations, together with the definition of the binarization degree used throughout this work, are provided in Appendix Sec.~\ref{sec::binar}. The architecture of the attention-based U-Net backbone is described in Appendix Sec.~\ref{sec::netdetai}, including the encoder--decoder structure, residual blocks, skip connections, and attention module employed in the inverse-design network.

\subsubsection*{Acknowledgments}
We thank Dr. Xinyu Nie, Dr. Shiyu Li and Mr. Jianwei Zhang from University of Southern California for valuable discussions.
We appreciate Professor Yonggang Shi from University of Southern California for providing the GPU workstations for the experiments. We are deeply grateful to the late Professor Chia Wei (Wade) Hsu from University of Southern California for his mentorship, guidance, and support, which laid an important foundation for this work. Though he is no longer with us, his wisdom, encouragement, and enduring legacy continue to guide and inspire us.

\subsubsection*{Funding} This work was supported by C.W. Hsu's National Science Foundation CAREER award ECCS-2146021. O.D. Miller's contributions were supported by a 2024 Sony Research Award and by the Simons Collaboration on Extreme Wave Phenomena Based on Symmetries (award no. SFI-MPS-EWP-00008530-09).


\bibliography{iclr2025_conference}
\bibliographystyle{unsrtnat}

\newpage

\appendix
\section{Comparison with existing deep-learning-based methods}
\label{sec::appcomexis}
To further evaluate the scalability and practical utility of the proposed dimension expansion network (DEN), we compare it with representative deep-learning-based nanophotonic design methods in terms of accessible design space, fabrication tolerance, structural complexity, and computational cost. The comparison includes both forward-simulation surrogate models, which predict optical responses from given structures, and inverse-design methods, which directly generate structures from target optical responses. While these approaches have demonstrated promising performance on a variety of nanophotonic devices, they often rely on large collections of simulated or optimized structures, resulting in substantial data-generation costs.

Figure~\ref{Figure_compare_existminstruc} compares the minimum normalized fabrication feature size, $\Delta V_{\min}/\lambda^d$, against the number of simulations or optimized structures used during training. Here, $\Delta V_{\min}$ denotes the minimum fabrication feature size within the accessible design space, normalized by the operating wavelength $\lambda$ and the number of independent design dimensions $d$. Larger values indicate that the generated structures remain manufacturable with relatively coarse fabrication resolution, whereas smaller values correspond to designs requiring finer fabrication precision.

A clear empirical trend can be observed across previously reported methods. In general, methods capable of accessing larger or more flexible design spaces tend to require increasingly fine fabrication features or substantially larger training datasets. This trend is captured by the regression curve and the associated $95\%$ confidence interval shown in Fig.~\ref{Figure_compare_existminstruc}. Consequently, improving design flexibility is often accompanied by reduced fabrication tolerance or increased computational cost.

In contrast, DEN occupies a distinct region of this performance--cost landscape. The proposed method lies substantially above the regression trend and outside the $95\%$ confidence interval, indicating a more favorable balance between fabrication tolerance and computational cost than existing approaches. For the metalens benchmark considered in this work, DEN achieves a minimum fabrication feature size of approximately $\lambda/15$ while requiring only about $8\times10^3$ physics-driven training simulations. This feature size is considerably larger than those reported by many existing deep-learning-based inverse-design methods, suggesting improved manufacturability without sacrificing design capability.

A broader comparison is summarized in Table~\ref{Table_comparison_others}, which includes representative approaches based on multilayer perceptrons (MLPs), convolutional neural networks (CNNs), transformers, tandem networks, generative adversarial networks (GANs), variational autoencoders (VAEs), and hybrid architectures. For each method, we compare the number of simulations or optimized structures used during training, the number of structural parameters, the physical device size, the minimum operating wavelength, the normalized designable volume, and the minimum fabrication feature size.

Several trends emerge from this comparison. First, forward-simulation approaches generally require large datasets to accurately approximate electromagnetic responses. Many methods rely on tens of thousands to millions of simulations, particularly when the parameter space becomes large or when broadband optical responses must be predicted. Although these methods provide efficient surrogates once trained, the initial data-generation cost can be substantial.

Second, many inverse-design approaches reduce the need for online optimization but often rely on collections of pre-generated optimized structures. Tandem networks require supervised pairs of structures and responses, while GAN- and VAE-based methods frequently require thousands to millions of training examples. In addition, generative approaches commonly produce distributions of candidate structures, requiring post-selection or further optimization to identify high-performance designs. As a result, the overall computational cost may significantly exceed the nominal size of the training dataset.

Third, existing methods often face a trade-off between design-space coverage, fabrication tolerance, and computational cost. Approaches capable of exploring large design spaces frequently require either extremely fine fabrication features, large training datasets, or both. For example, several GAN-based and VAE-based methods operate in high-dimensional parameter spaces but rely on extensive collections of pre-generated structures and often achieve minimum feature sizes substantially smaller than those reported here.

In contrast, DEN directly performs physics-based inverse design without requiring either pre-generated simulation datasets or optimized structures. The network is trained solely through electromagnetic feedback and therefore learns a target-conditioned design manifold rather than a supervised mapping from previously generated examples. Consequently, the computational effort is invested in learning the design process itself instead of constructing training data.

This advantage is reflected in both benchmarks considered in this work. The proposed metalens benchmark achieves a normalized designable volume of approximately $40$, one of the largest values reported among the methods surveyed here, while maintaining a comparatively large minimum fabrication feature size of approximately $\lambda/15$. Similarly, the proposed Y-splitter benchmark addresses a 900-parameter inverse-design problem using only 3750 training simulations and achieves a minimum fabrication feature size of approximately $\lambda/28$. Notably, the number of structural parameters in the Y-splitter benchmark exceeds that of many previously reported photonic inverse-design studies, yet the required computational cost remains relatively modest.

Taken together, these results indicate that DEN scales favorably with both structural dimensionality and accessible design space while preserving fabrication-friendly geometries. Rather than relying on large collections of optimized examples, the proposed method amortizes optimization across the target space through direct electromagnetic training. As a result, DEN simultaneously achieves large designable volumes, relatively coarse fabrication requirements, and low simulation costs.

Overall, the proposed method occupies a distinct region of the performance--cost landscape compared with existing deep-learning-based nanophotonic design approaches. By combining fully unsupervised training with direct electromagnetic optimization, DEN substantially reduces the data-generation burden while enabling large accessible design spaces and fabrication-tolerant structures. These results suggest that DEN overcomes a key limitation of many existing deep-learning-based inverse-design frameworks and provides a favorable balance between computational cost, design flexibility, and manufacturability.

\begin{figure}[hbt]
    \centering
    \includegraphics[width=0.6\textwidth]{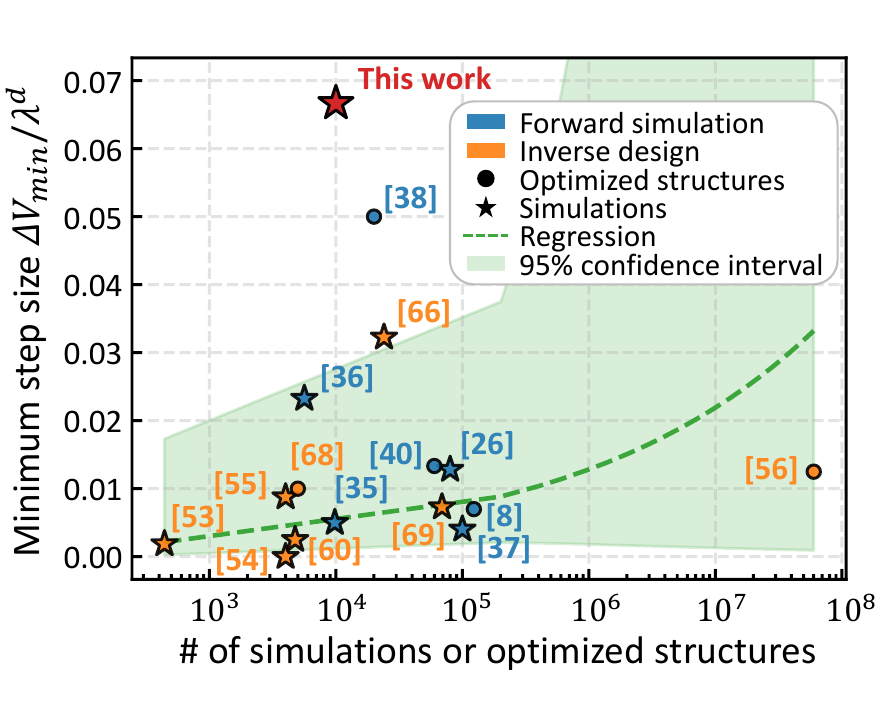}
    \caption{\textbf{Comparison with existing metalens design methods.} We compared the minimum design step size (in $\Delta V_{min}/(\lambda)^d$, where $\Delta V_{min}$ is the minimum step size measured in units consistent with the design dimensionality and $d$ is the number of independent design axes) and the number of optimized structures or simulations used to generate training data and train the network. ``Forward simulation'' are methods that only use deep learning to predict the forward simulation results of given structures, and ``Inverse design'' are methods that directly output the designed structures. ``Optimized structures'' refer to datasets in which each training structure is produced through optimization, typically requiring many more simulations than the number of training samples, whereas ``simulations'' denote datasets generated by random sampling with only one simulation per data point. The regression line and the 95\% confidence interval are also shown. The regression follows a correlation of the form $y = a x^{b}$.}
    \label{Figure_compare_existminstruc}
\end{figure}

\begin{landscape}
\begin{table}[tb]
\centering
\caption{Representative deep learning models for nanophotonics devices forward simulation and inverse design.}
\label{Table_comparison_others}
\begin{threeparttable}
\resizebox{1.45\textwidth}{!}{
\begin{tabular}{l l l c c c c c c}
\toprule
\textbf{Structure} & \textbf{Network} & \textbf{Ref.} &
\textbf{Sims (structs) in training} & \textbf{Struct. params} & \textbf{Struct. size (nm)} &
\textbf{$\lambda_{\min}$ (nm)} & \textbf{Designable volume} &
\textbf{Min. step size (nm)}\\
\midrule
\multicolumn{9}{c}{\textbf{Forward simulation}} \\
\midrule
Meta atom & MLP & \cite{Zhang2024_Vortex} & 9,800 & 4 (26)\tnote{*} & 500$\times$500 & 1,000 & 1 & 5 ($\lambda/200$)\\
Meta atom & MLP & \cite{mlpdnnatom} & 5,630 & 4 (25)\tnote{*} & 320$\times$320 & 428 & 2.24 & 10 ($\lambda/43$)\\
Meta atom & MLP & \cite{Zhelyeznyakov2021_ScattererField} & (123,210) & 1 & 443$\times$443 & 633 & 1.96 & 4.43 ($\lambda/143$)\tnote{***}\\
Meta atom & CNN & \cite{An2022_CouplingCNN} & 100,000 & 2 & 800$\times$800 & 1,550 & 1.1 & 6.25 ($\lambda/248$)\\
Meta atom & Transformer & \cite{Wang2025_DualLayer} & 80,000 & 196 & 300$\times$300 & 520 & 1.33 & 20 ($\lambda/78$)\\
Meta atom & Transformer & \cite{transformeratomforw} & (20,000) & 400 & 420$\times$420 & 400 & 4.4 & 20 ($\lambda/20$)\\
Meta atom & Mixed & \cite{Wang2025_KAN} & (60,000) & 256 (201)\tnote{*} & 700$\times$640 & 1,500 & 0.80 & 20 ($\lambda/75$)\\
Metasurface spectral filters & MLP & \cite{mlpatomforward} & 40,000 & 17 & 1,000$\times$1,000 & 1,000 & 4 & 50 ($\lambda/20$)\\
Square plasmonic switch & MLP & \cite{mlpatomsixout} & (18,432) & 6 (2)\tnote{*} & 1,000$\times$1,000 & 1,000 & 4 & 3 ($\lambda/333$)\\
Nanoparticles & MLP & \cite{Peurifoy2018_NanoparticleANN} & 50,000 & 8 (7)\tnote{*} & 1,120 & 400 & 5.6 & 1 ($\lambda/400$)\tnote{***}\\
Y-splitter & MLP & \cite{Tahersima2019_Splitter} & (20,000) & 400 & 2,600$\times$2,600 & 1,450 & 12.86 & 90 ($\lambda/16$)\\
\midrule
\multicolumn{9}{c}{\textbf{Inverse design}} \\
\midrule
Meta atom & MLP & \cite{Gu2021_Bifocal} & 441 & 2 & 460$\times$460 & 532 & 2.99 & 1 ($\lambda/532$)\tnote{***}\\
Meta atom & MLP & \cite{mlpatominver} & 4,000 & 5 & 100,000$\times$100,000 & 300,000 & 0.44 & 1 ($\lambda/300,000$)\tnote{***}\\
Meta atom & Tandem MLP & \cite{Qiu2024_GSST} & 4,758 & 2 & 360 & 8,000 & 0.09 & 20 ($\lambda/400$)\tnote{***}\\
Meta atom & MLP & \cite{dnnatom} & 4,000 & 2 & 2,200,000$\times$7,000,000 & 2,855,166 & 7.56 & 25,000 ($\lambda/114$)\\
Meta atom & Transformer & \cite{Yan2025_MetasurfaceViT} & 60,000,000 & 6 & 300$\times$300 & 400 & 2.25 & 5 ($\lambda/80$)\tnote{***}\\
Meta atom & CDCGAN & \cite{cdcganatom} & (5,000) & 144 & 7,200,000$\times$7,200,000 & 30,000,000 & 0.23 & 300,000 ($\lambda/100$)\\
Meta atom & VAE & \cite{Xiong2024_Multiwavelength} & 24,000 & 4,096 & 640$\times$640 & 1,000 & 1.64 & 32 ($\lambda/31$)\\
Meta atom & GAN & \cite{An2021_MultifunctionalGAN} & 69,000 & 1,024\tnote{**} & 2,800$\times$2,800 & 6,000 & 0.87 & 43.75 ($\lambda/137$)\\
Complex amplitude filter & MLP & \cite{Ou2025_FocusEngineering} & 10,000 & 20 & 250 & 532 & 0.94 & 2 ($\lambda/266$)\\
Chiral metamirror & VAE & \cite{Ma2019_ProbabilisticVAE} & 20,000 & 100 & 2,000$\times$2,000 & 3,000 & 1.78 & 31.25 ($\lambda/96$)\\
Periodic atoms & cGAN & \cite{Mall2020_CyclicalDL} & (1,500) & 4,096\tnote{**} & 230$\times$230 & 400 & 1.32 & 3.59 ($\lambda/111$)\\
Periodic atoms & GAN & \cite{Xia2025_DeepLearningHybrid} & (15,640) & 256 & 70,000$\times$70,000 & 100,000 & 1.96 & 4,000 ($\lambda/25$)\\
Plasmonic metasurface & CNN & \cite{cnnatom} & 25,000 & 6 & 700$\times$700 & 600 & 5.44 & 1 ($\lambda/600$)\\
Plasmonic metasurface & MLP & \cite{mlpatomplas} & 15,000 & 8 & $\sim$400$\times$400\tnote{***} & 600 & 1.78 & Not given\\
Silver antennae & GAN & \cite{So2019_cDCGAN} & 10,150 & 4,096\tnote{**} & 500$\times$500 & 600 & 2.78 & 7.81 ($\lambda/216$)\\
Chiral metamaterial & Tandem MLP & \cite{Ma2018_Chiral} & 25,000 & 5 & $\sim$2,500$\times$2,500$\times$1,000\tnote{***} & 3,747 & 0.95 & 1 ($\lambda/3747$)\\
Metagrating & MLP & \cite{mlpgrating} & 50,000 & 9 & 4,500 & 1,000 & 9 & 30 ($\lambda/33$)\\
Metagrating & MLP & \cite{Adibi_2020_npjC} & 4,000 & 10 & 200$\times$2,000 & 1,500 & 0.71 & 10 ($\lambda/150$)\tnote{***}\\
Metagrating & GAN & \cite{Jiang2019_MetagratingGAN} & (5,000) + $5,000 \times n_{test}$ & 32,768 & 500$\times$1,220 & 1,000 & 2.44 & 3.90 ($\lambda/256$)\\
Metagrating & GAN & \cite{Wen2020_RobustGAN} & 7,008,000 & 8,192 & 650$\times$2,269 & 1,300 & 3.49 & 17.73 ($\lambda/73$)\\
Patterned 2D slit array & GAN & \cite{ganpatter} & (24,350) & 800 & 40 & 400 & 0.2 & 2 ($\lambda/200$)\\
Thin film & Tandem MLP & \cite{thinmlp} & 550,000 & 20 & 2,400 & 400 & 12 & 1 ($\lambda/400$)\\
Grating coupler & Tandem MLP & \cite{Xu2021_TandemImproved} & 84,480 & 120 & 27,600$\times$220 & 1,300 & 14.37 & 10 ($\lambda/130$)\\
Y-splitter & GAN & \cite{Tang2020_IntegratedGAN} & (16,000) + training & 400 & 2,250$\times$2,250 & 1,250 & 12.96 & 42 ($\lambda/30$)\\
Y-splitter & GAN & \cite{ganysplit} & (10,000) & 400 & 2,000$\times$2,000 & 1,200 & 11.11 & 100 ($\lambda/12$)\\
\midrule
Y-splitter & DEN & This work & 0 + 3750 in training & 900 & 1,660$\times$1,660 & 1,550 & 4.59 & 55 ($\lambda/28$)\\
Metalens & DEN & This work & 0 + $100\times \frac{S_{focal\_region}}{(\lambda/2)^2}$ in training & 301 & 10,000 & 500 & 40 & 33.33 ($\lambda/15$)\\
\bottomrule
\end{tabular}}
\begin{tablenotes}[flushleft]\footnotesize
\item[*] Numbers in parentheses denote output field parameters (e.g., phase response in \cite{Zhang2024_Vortex}).
\item[**] GAN models can handle large parameter spaces but often reproduce training data patterns.
\item[***] These numbers are not reported. They are calculated from the results section.
\end{tablenotes}
\end{threeparttable}
\end{table}
\end{landscape}

\section{Network structures}
\label{sec::netdetai}

\begin{figure}[hbt]
    \centering
    \includegraphics[width=\textwidth]{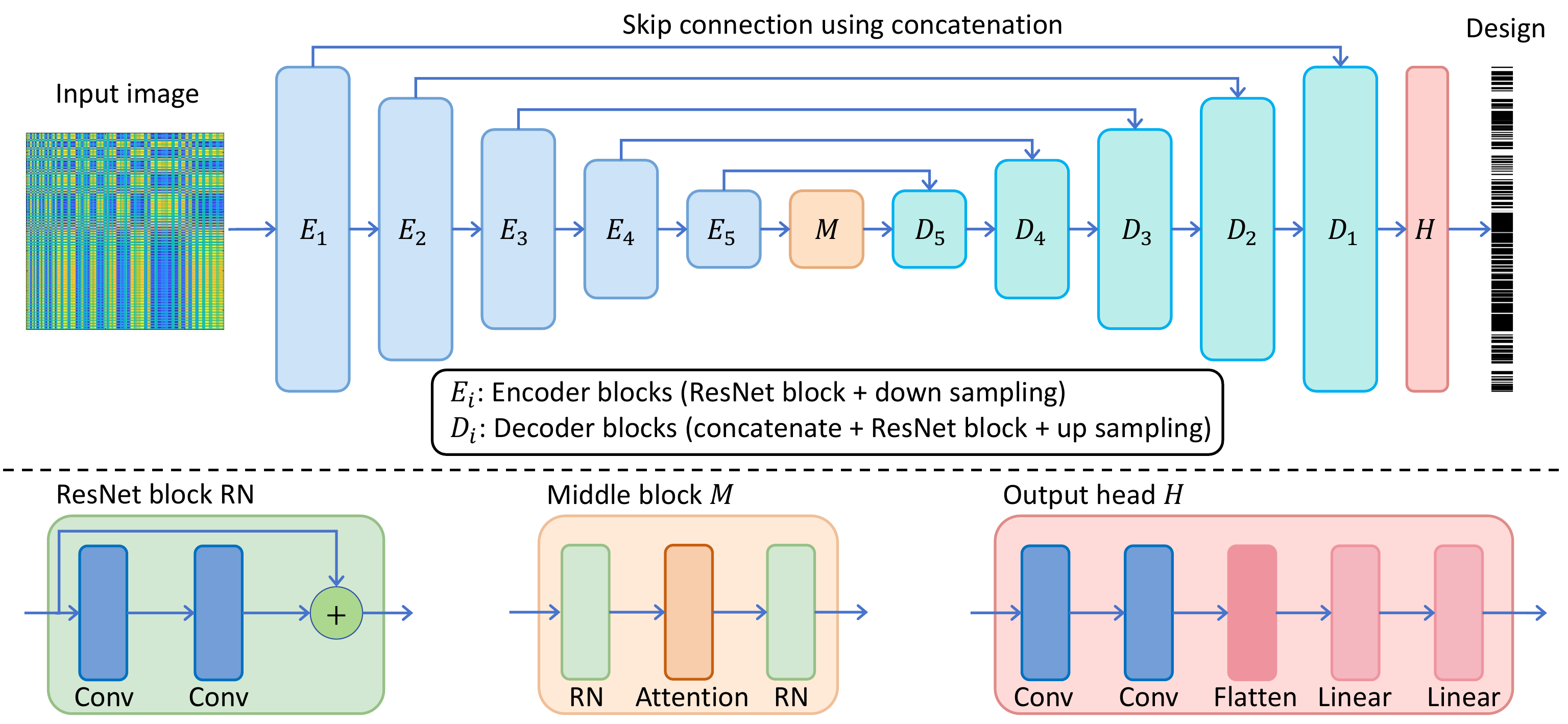}
    \caption{\textbf{Network structure for design generation from the input image.} The proposed framework employs an attention-based U-Net architecture to generate nanophotonic structures from the dimension-expanded inputs. The network consists of encoder blocks ($E_i$), decoder blocks ($D_i$), skip connections, a bottleneck attention module ($M$), and an output head ($H$). Residual connections are incorporated within each block to facilitate optimization and improve feature propagation.}
    \label{Figure_unet_struc}
\end{figure}

We employ an attention-based U-Net architecture to map the dimension-expanded target representation to the corresponding refractive-index distribution. This architecture is particularly suitable for the proposed framework because both the input and output possess explicit spatial structure. The dimension-expansion module transforms the target parameters into a two-dimensional feature map, while the desired output is a spatial refractive-index profile. Consequently, the inverse-design problem can be formulated as an image-to-image mapping task, for which encoder--decoder architectures have proven highly effective~\cite{unet}.

As shown in Fig.~\ref{Figure_unet_struc}, the network consists of three main components: an encoder, a decoder, and skip connections between corresponding stages. The encoder progressively extracts hierarchical features through successive down-sampling operations, allowing local information to be aggregated into increasingly global representations. The decoder then reconstructs the spatial refractive-index distribution from these latent features through up-sampling operations. Skip connections directly transfer high-resolution information from the encoder to the decoder, enabling the network to preserve fine structural details while maintaining global consistency.

Each encoder and decoder stage contains residual convolutional blocks~\cite{resn}, which improve optimization stability and facilitate information propagation through the network. The residual formulation is particularly beneficial in the present application because the generated structures often contain both fine-scale fabrication features and larger-scale geometric patterns that must be represented simultaneously.

The down-sampling operations progressively enlarge the effective receptive field of the network. For example, in the Y-splitter design task, where the pixel size is $55\times55~\mathrm{nm}^2$, a $3\times3$ convolution kernel initially covers a physical region of $165\times165~\mathrm{nm}^2$. After four down-sampling stages, the same kernel effectively captures information over a region of approximately $2640\times2640~\mathrm{nm}^2$, comparable to the size of the entire design domain. This enables the network to capture long-range interactions between different regions of the structure.

To further model global dependencies, we introduce an attention module in the bottleneck layer of the U-Net~\cite{vaswani2017attention}. Unlike standard convolutions, which primarily aggregate information locally, the attention mechanism allows distant spatial locations to interact directly. This capability is particularly important for nanophotonic inverse design, where modifications in one region of the structure can significantly influence the optical response generated by distant regions. The attention block therefore complements the convolutional layers by capturing long-range correlations that are difficult to represent using local operations alone.

The final output head converts the decoded feature representation into the refractive-index distribution of the device. Two convolutional layers are first used to refine the reconstructed spatial features, followed by fully connected layers that generate the final structure representation. The resulting output is subsequently binarized using the differentiable procedures described in Sec.~\ref{sec_met}, ensuring compatibility with practical fabrication constraints.

It is important to note that the primary contribution of this work is not the specific network architecture itself, but rather the proposed dimension-expansion framework that enables effective conditioning of high-dimensional inverse-design problems. The attention-based U-Net serves as a representative backbone architecture, while the experiments in Sec.~\ref{sec_analy_den} demonstrate that the benefits of DEN generalize to substantially different architectures, including multilayer perceptrons.

\section{Binarization method}
\label{sec::binar}

To ensure fabrication compatibility, the refractive-index distribution generated by the network is gradually driven toward binary values during training. A direct thresholding operation would introduce discontinuities and make the optimization problem non-differentiable, preventing the use of gradient-based learning. Instead, we employ differentiable binarization strategies that allow the network to explore continuous design spaces during the early stages of optimization while progressively converging toward binary structures as training proceeds.

For the metalens design task, we adopt an imaginary-penalty approach. The refractive index is expressed as
\begin{equation}
n = n_r + i\,\alpha(n_r)\,\gamma(t),
\label{equation_metalens}
\end{equation}
where $n_r$ is the real-valued refractive index predicted by the network,
\begin{equation}
n_r
=
\mathrm{clip}(x+1.7,\; 1.0,\; 2.4),
\end{equation}
and
\begin{equation}
\alpha(n_r)
=
\left|
|n_r-1.7|-0.7
\right|.
\end{equation}
The penalty function $\alpha(n_r)$ vanishes at the desired binary refractive-index values ($n_r=1.0$ and $n_r=2.4$) and reaches its maximum near the midpoint between them. Consequently, intermediate refractive-index values experience stronger attenuation than values close to the binary states. The weighting factor $\gamma(t)$ is gradually increased during training, progressively strengthening the penalty and driving the refractive-index distribution toward binary solutions. Because the formulation remains continuous and differentiable, gradients can still be propagated through the electromagnetic simulation during optimization.

For the Y-splitter design task, we employ a sigmoid-based binarization scheme. The refractive index is parameterized as
\begin{equation}
y
=
1.0
+
\frac{1}{1+e^{-wx}}
\times 2.5,
\label{equation3}
\end{equation}
where $x$ is the network output and $w$ controls the sharpness of the transition between air ($n=1.0$) and silicon ($n=3.5$). During training, the parameter $w$ is gradually increased. As a result, the sigmoid function evolves from a smooth continuous mapping to an increasingly sharp approximation of a binary threshold. This allows the network to initially explore a larger design space before progressively enforcing binary structures.

Although the two formulations are different, they share the same underlying principle. Both methods preserve differentiability during the early stages of optimization and gradually increase the strength of the binarization constraint as training proceeds. This continuation strategy improves optimization stability, avoids abrupt changes in the loss landscape, and enables the network to discover high-performance structures before converging to fabrication-compatible binary solutions.

Importantly, the proposed framework is not tied to a specific binarization strategy. The imaginary-penalty and sigmoid-based formulations used here serve as two representative examples for the metalens and Y-splitter design tasks, respectively. More generally, DEN can be combined with a broad class of differentiable fabrication-aware constraints while preserving end-to-end differentiability throughout the optimization process. This flexibility allows the framework to adapt to different device geometries, material systems, and fabrication requirements without modifying the underlying dimension-expansion network.

To quantify the degree of binarization of the generated structures, we define the binarization degree as

\begin{equation}
B
=
\frac{1}{N}
\sum_{i=1}^{N}
\left|
\frac{2(n_i-n_c)}
{n_{\max}-n_{\min}}
\right|,
\end{equation}
where $N$ is the total number of pixels in the design region, $n_i$ denotes the refractive index of the $i$-th pixel, and
\begin{equation}
n_c
=
\frac{n_{\max}+n_{\min}}{2}
\end{equation}
is the midpoint between the two target refractive-index values. The metric is normalized such that $B=1$ corresponds to a perfectly binary structure and $B=0$ corresponds to a completely intermediate-valued structure. Throughout this work, the generated designs typically achieve binarization degrees above $0.95$, indicating that the resulting structures are highly compatible with binary fabrication processes.

\section{Training data selection for Y-splitter}
\label{sec::tradatyspl}

\begin{figure}[hbt]
    \centering
    \includegraphics[width=\textwidth]{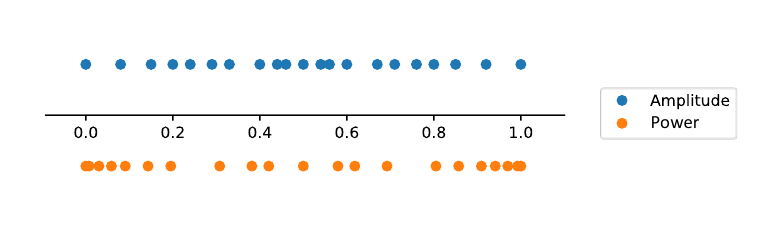}
    \caption{\textbf{Training data selection for the Y-splitter design task.} We firstly randomly selected 9 amplitude splitting ratios between 1\% : 99\% and 49\% : 51\%. Then the two boundary cases, 0\% : 100\% and 50\% : 50\%, are added. We then calculate the corresponding power splitting ratios and use these ratios as the training data. The symmetric data of the selected power splitting ratios are also used in the training.}
    \label{Figure_training_data}
\end{figure}

The Y-splitter design task aims to generate structures with arbitrary power splitting ratios between the two output ports. In practice, the corresponding structural variations are not uniformly distributed across the target space. Designs associated with extreme splitting ratios (e.g., $T_1 \rightarrow 0\%$ or $T_1 \rightarrow 100\%$) typically differ more significantly from one another than those associated with intermediate ratios. Consequently, a uniform sampling strategy in power space may provide insufficient coverage near the boundaries of the target range.

To address this issue, we construct the training targets in amplitude space rather than directly in power space. Specifically, we first randomly select nine amplitude ratios in the range from $1\%$ to $49\%$. The corresponding power splitting ratios are then obtained through the quadratic relationship
\begin{equation}
    T_1 = A^2,
\end{equation}
where $A$ denotes the selected amplitude ratio. Because of this nonlinear transformation, the resulting power-ratio samples become naturally concentrated near the extreme values of the target range. This produces a denser sampling of the regions where the structural response changes most rapidly.

The selected amplitude ratios and their corresponding power ratios are illustrated in Fig.~\ref{Figure_training_data}. While the amplitude samples are distributed approximately uniformly, the transformed power-ratio samples become increasingly dense near $0\%$ and $100\%$, providing improved coverage of the extreme splitting conditions.

To ensure symmetry between the two output ports, we further augment the training set by including the complementary power ratios $100\%-T_1$. In addition, three anchor points corresponding to the two extreme cases and the balanced splitter are explicitly included:
\begin{equation}
    T_1 = 0\%, \qquad
    T_1 = 50\%, \qquad
    T_1 = 100\%.
\end{equation}
These targets guarantee that the network directly observes the most important boundary and symmetry conditions during training.

As a result, the entire continuous target space is represented using only 21 carefully selected training samples. Despite this small training set, the resulting model is able to accurately interpolate between the sampled targets and generate high-performance designs across the full range of splitting ratios, as demonstrated in Sec.~\ref{sec::yspli}.

\section{Examples of DEN metalens}
\label{sec::examp}

To further illustrate the design capability of the proposed DEN, Fig.~\ref{Figure_example_stu} presents representative metalens structures generated for different target focal positions. The examples correspond to the largest focal-region size considered in this work, $D = 2.67~\mu\mathrm{m}$, at the operating wavelength $\lambda = 500~\mathrm{nm}$. The target focal points span a two-dimensional region with varying horizontal ($f_x$) and vertical ($f_y$) coordinates, resulting in a diverse set of design objectives within the same design space.

For each target location, DEN directly predicts a corresponding binary-like refractive-index distribution without requiring any iterative optimization during inference. Owing to the mirror symmetry of the design problem, focal points located symmetrically with respect to the optical axis correspond to two symmetry-related candidate structures. For every target focal position shown in Fig.~\ref{Figure_example_stu}, both symmetry-related candidates are evaluated, and the structure with the higher focal intensity is selected for visualization. This comparison ensures that the reported examples represent the best-performing design among the symmetry-equivalent solutions.

The resulting structures exhibit clear target-dependent variations and produce distinct optical responses. The corresponding electric-field intensity distributions demonstrate that the optical energy is consistently concentrated near the desired target locations marked by the red crosses. Across all examples, the focal spots closely follow the prescribed target coordinates, indicating that the learned mapping remains effective throughout the entire focal region.

As the target focal coordinates vary across the focal space, the focal spots move accordingly in both the horizontal and vertical directions while remaining well localized. Neighboring target positions produce related optical responses, whereas larger changes in focal position lead to increasingly different field distributions and structural patterns. These observations indicate that the network has learned a meaningful relationship between target coordinates and nanophotonic structures rather than converging to a single generic solution.

The field distributions further demonstrate robust focusing performance over a broad range of focal locations. Even when the target positions are located near the boundaries of the target focal region, the generated structures maintain strong field confinement and relatively low background scattering. This behavior is particularly notable because the focal region considered here spans multiple wavelengths in both directions, requiring the network to generate a wide variety of structures capable of redirecting light over a broad angular range.

Importantly, all of these designs are generated using the same trained network. Unlike conventional adjoint optimization, which requires an independent optimization process for each focal position, DEN learns a shared target-conditioned mapping that can efficiently produce many high-performance designs across the entire focal space. These examples therefore provide qualitative evidence that the proposed framework successfully captures the relationship between focal coordinates and nanophotonic structures over a large continuous target space.

Together with the quantitative results presented in Fig.~\ref{Figure_compare_adjoi}(c) of the main text, the examples in Fig.~\ref{Figure_example_stu} demonstrate that DEN can efficiently generate diverse, fabrication-compatible, and high-performance free-form metalens designs while substantially reducing the computational cost associated with repeated per-target optimization.

\begin{figure}
    \centering
    \includegraphics[width=\linewidth]{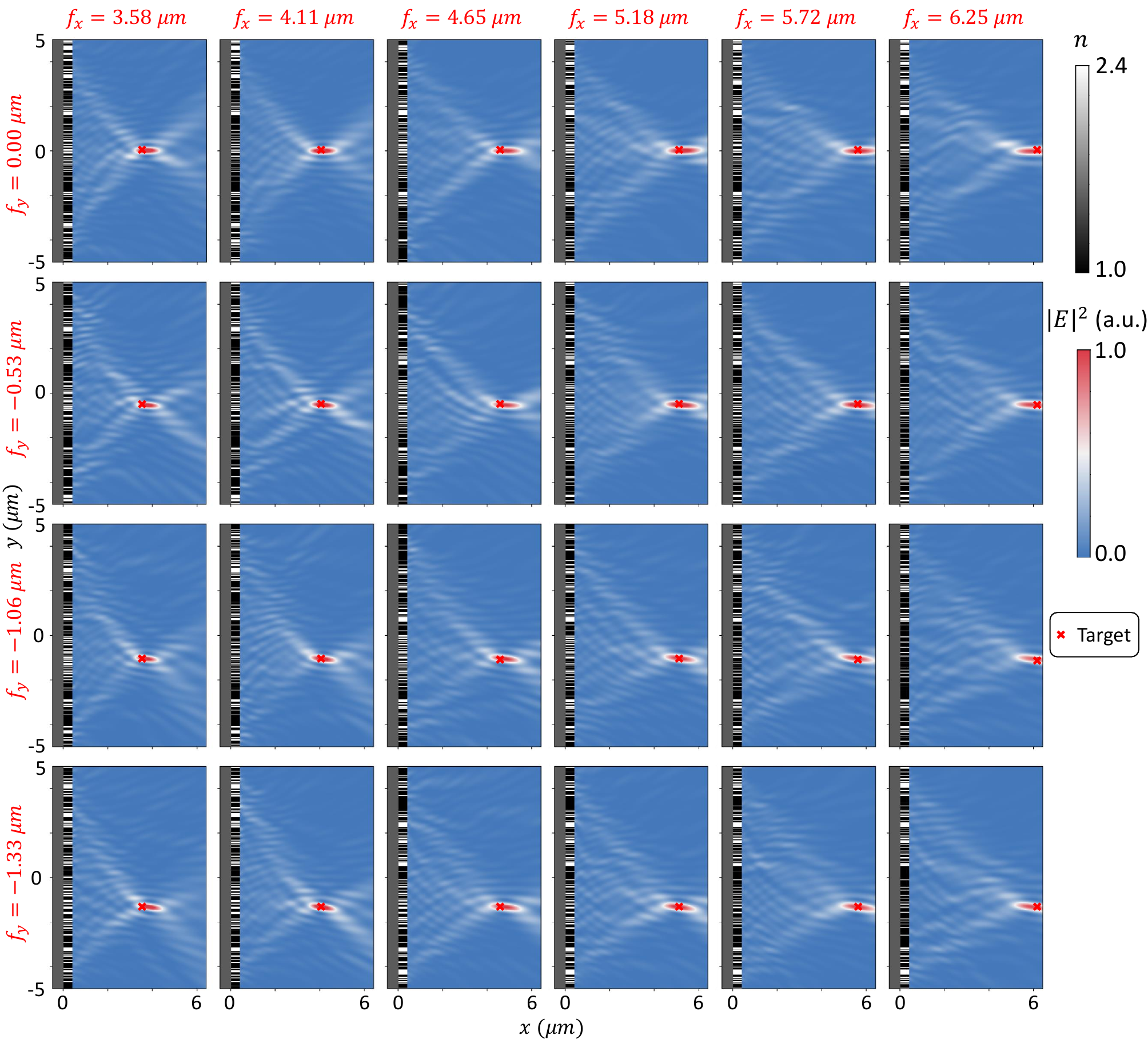}
    \caption{\textbf{Examples of metalens structures generated by DEN for $D = 2.67~\mu \mathrm{m}$ and $\lambda = 500~\mathrm{nm}$.} The target focal points span different horizontal ($f_x$) and vertical ($f_y$) coordinates within the focal region. For each target location, DEN directly predicts a corresponding binary-like refractive-index distribution and focuses light to the desired position. The red crosses indicate the target focal points. The generated structures produce well-confined focal spots with high intensity and low background scattering across the full target region.}
    \label{Figure_example_stu}
\end{figure}

\section{Scalability to dense per-pixel target sampling}
\label{sec::pixeldes}

\begin{table}[tb]
\centering
\caption{Performance comparison between sparse target sampling and dense per-pixel target sampling under different target-region sizes $D$.}
\label{table_efficiency_binary}
\begin{tabular}{c c c c}
\hline
$D$ ($\mu\mathrm{m}$) & \# designs & Focal intensity (a.u.) & Degree of binarization \\
\hline
\multirow{2}{*}{0.80} 
& 625 ($25 \times 25$) & $0.86 \pm 0.07$ & $0.96 \pm 0.01$ \\
& 25 ($5 \times 5$)  & $0.91 \pm 0.05$ & $0.97 \pm 0.01$ \\

\hline
\multirow{2}{*}{1.33} 
& 1681 ($41 \times 41$) & $0.86 \pm 0.08$ & $0.95 \pm 0.01$ \\
& 36 ($6 \times 6$)  & $0.90 \pm 0.04$ & $0.97 \pm 0.01$ \\
\hline
\multirow{2}{*}{2.67} & 6561 ($81 \times 81$) & $0.81 \pm 0.07$ & $0.95 \pm 0.01$ \\
& 121 ($11 \times 11$)  & $0.86 \pm 0.06$ & $0.97 \pm 0.01$ \\
\hline
\end{tabular}
\end{table}

\begin{figure}
    \centering
    \includegraphics[width=0.6\linewidth]{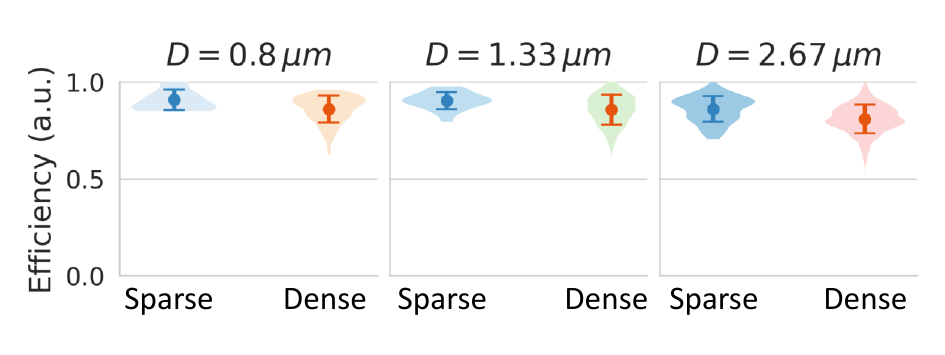}
    \caption{\textbf{Scalability to dense per-pixel target sampling.} Distribution comparisons of the focal intensity for sparse and dense target-space sampling under different target-region sizes $D$. The light violins show the distributions of the generated structures, while the markers and error bars indicate the mean and standard deviation, respectively. Even when the number of target focal points increases from sparse grids ($5\times5$, $6\times6$, and $11\times11$) to dense per-pixel sampling ($25\times25$, $41\times41$, and $81\times81$), the performance degradation remains moderate, demonstrating the strong scalability of DEN under a fixed training budget.}
    \label{Figure_compare_pix}
\end{figure}

To evaluate the scalability of DEN under increasingly dense target-space coverage, we compare sparse target sampling and dense per-pixel target sampling using the same training budget. For each target-region size $D$, the sparse and dense configurations are trained using the same number of optimization iterations as those reported in Fig.~\ref{Figure_compare_adjoi}(c) of the main text. Consequently, any performance differences arise solely from the increased density of target sampling rather than additional computational resources.

The quantitative results are summarized in Table~\ref{table_efficiency_binary}. As the target discretization becomes progressively denser, the number of target focal points increases dramatically while the accessible focal region remains unchanged. For example, for the largest target region ($D=2.67~\mu\mathrm{m}$), the number of focal positions increases from $121$ ($11\times11$) to $6561$ ($81\times81$), corresponding to more than a fifty-fold increase in target-space coverage. Despite this substantial increase in the number of target focal positions, the average focal intensity decreases only from $0.86\pm0.06$ to $0.81\pm0.07$, corresponding to a reduction of approximately $6\%$.

A similar trend is observed for the smaller focal regions. For $D=0.80~\mu\mathrm{m}$, increasing the number of target focal points from $25$ to $625$ reduces the average focal intensity from $0.91\pm0.05$ to $0.86\pm0.07$. For $D=1.33~\mu\mathrm{m}$, increasing the target count from $36$ to $1681$ decreases the average focal intensity from $0.90\pm0.04$ to $0.86\pm0.08$. In all cases, the performance degradation remains modest relative to the large increase in the number of target focal positions represented by the model.

Importantly, the degree of binarization remains consistently high across all experiments. Even under dense per-pixel target sampling, the generated structures maintain binarization levels of approximately $0.95$-$0.97$, indicating that the increased target-space coverage does not compromise fabrication compatibility. These results suggest that DEN can simultaneously maintain optical performance and manufacturability even when required to represent a much denser set of target conditions.

To further visualize this behavior, Fig.~\ref{Figure_compare_pix} shows the distributions of focal intensities for both sparse and dense target sampling. The violin plots illustrate the performance distributions across all generated structures, while the markers and error bars denote the corresponding means and standard deviations. For all three target-region sizes, the sparse and dense distributions exhibit substantial overlap and only a modest shift in their mean values. The overall distribution shapes remain similar despite the large increase in the number of target focal positions.

These observations indicate that DEN scales favorably with target-space coverage. For the largest focal region, the number of target focal positions increases by more than fifty times, while the average focal intensity decreases by only approximately $6\%$. Similar behavior is observed for the other focal-region sizes. This behavior highlights a key advantage of the proposed framework: a single network can efficiently represent and generate a very large number of target-dependent designs while maintaining high performance under a fixed training budget. The results further suggest that the learned target-conditioned mapping generalizes effectively across densely sampled target spaces rather than relying on a sparse set of isolated design solutions.

\section{Additional analysis of dimension expansion}
\label{sec::addana}

\begin{figure}[t]
    \centering
    \includegraphics[width=\linewidth]{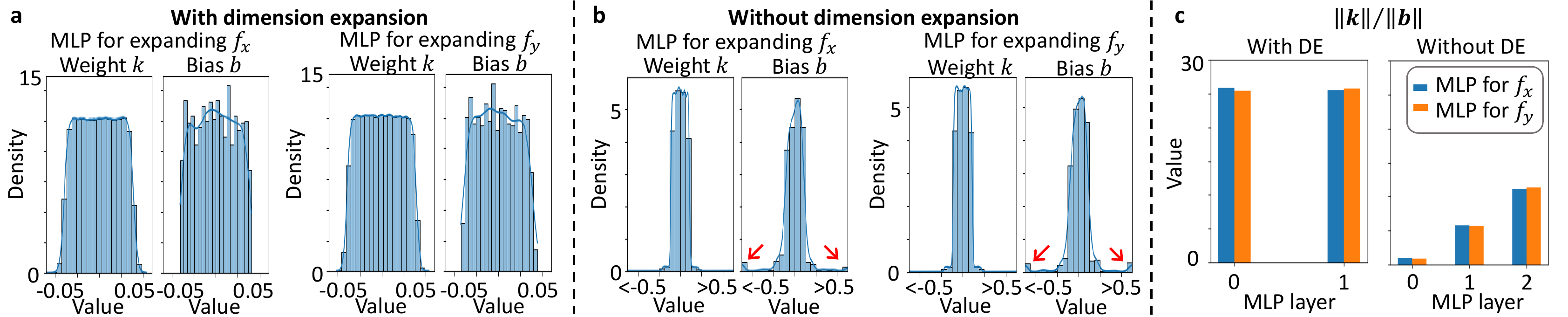}
    \caption{
    \textbf{Additional analysis of the projection MLPs with and without dimension expansion (DE).} \textbf{(a)} Distributions of MLP weights $k$ and biases $b$ with DE. The weights and biases remain well balanced across different layers. \textbf{(b)} Distributions of MLP weights and biases without DE. The bias distributions exhibit long tails (red arrows), indicating increased reliance on large biases rather than informative target-dependent features. \textbf{(c)} Comparison of the weight-to-bias ratio $||k||/||b||$ across different MLP layers. Networks with DE maintain substantially larger weight-to-bias ratios, suggesting stronger feature-dependent representations and reduced mode-collapse behavior.
    }
    \label{Figure_explain_exten}
\end{figure}

To further investigate the role of dimension expansion, we analyze the parameter distributions of the projection MLPs used to process the target coordinates. Figure~\ref{Figure_explain_exten} compares the learned weights and biases of the projection MLPs with and without dimension expansion (DE). These analyses provide additional insight into how DE influences the learned representations and why it improves the performance of the inverse-design network.

Figure~\ref{Figure_explain_exten}(a) shows the parameter distributions when DE is employed. For both projection MLPs, the weight and bias distributions remain relatively balanced and broadly distributed around zero. The learned parameters make effective use of both weights and biases, indicating that the output representations are primarily determined by the input-dependent transformations learned by the network. As a result, different target coordinates are mapped to distinct high-dimensional representations, providing rich conditioning information for the subsequent inverse-design network.

In contrast, Fig.~\ref{Figure_explain_exten}(b) shows the parameter distributions obtained without DE. While the weight distributions remain relatively concentrated, the bias distributions develop pronounced long tails, as indicated by the red arrows. This behavior suggests that the network increasingly relies on large bias terms rather than input-dependent feature transformations. Consequently, the generated representations become less sensitive to changes in the target coordinates, making it more difficult for the network to distinguish between nearby targets.

To quantify this effect, Fig.~\ref{Figure_explain_exten}(c) compares the ratio
\begin{equation}
    \frac{|k|}{|b|},
\end{equation}
where ($|k|$) and ($|b|$) denote the norms of the weights and biases, respectively. With DE, the ratio remains consistently large across different layers, indicating that the learned representations are dominated by feature-dependent transformations. Without DE, however, the ratio decreases substantially, reflecting the increasing influence of bias terms. This trend is consistent for both projection MLPs corresponding to ($f_x$) and ($f_y$).

These observations provide additional evidence that dimension expansion improves the conditioning of the inverse design problem. By transforming low-dimensional target coordinates into richer high-dimensional representations, DE encourages the network to rely on informative target-dependent features rather than bias-dominated mappings. This behavior improves representation diversity, enhances sensitivity to target variations, and helps prevent the mode-collapse phenomenon observed in the main text.

\section{Mathematical analysis of dimension expansion}
\label{sec::matheanaly}

In this section, we analyze the Fourier encoding and outer-product lifting used in the proposed dimension expansion framework, and investigate how they influence the geometry of the conditioning representations. The resulting analysis provides a theoretical perspective on the improved target separability, increased effective rank, and reduced mode-collapse behavior observed in Sec.~\ref{sec_analy_den}.
 
Let the target parameters be denoted by
\begin{equation}
\mathbf{T}=(x_1,x_2),
\end{equation}
where, following Sec.~\ref{sec::dimenexpan}, each coordinate is measured in
units of the target-grid index, so that $x\in[0,N_x]$ with $N_x$ the number of
discrete target samples along one dimension. Each coordinate is mapped into
the Fourier feature representation
\begin{equation}
I(x)
=
\left[
\sin(\omega_1x),\cos(\omega_1x),
\ldots,
\sin(\omega_{N_i}x),\cos(\omega_{N_i}x)
\right]^{\top}
\in\mathbb{R}^{2N_i},
\qquad
\omega_i=N_x^{-2i/N_i},
\label{eq:app_phi}
\end{equation}
so that the frequencies form a geometric ladder decreasing from
$\omega_{\max}=N_x^{-2/N_i}$ to $\omega_{\min}=N_x^{-2}$.
The encoded features are processed by two projection networks,
$\mathbf f_1=\mathrm{MLP}_1(I(x_1))$ and
$\mathbf f_2=\mathrm{MLP}_2(I(x_2))$, both with output dimension $2N_i$,
and combined through the outer product
$\mathbf g=\mathbf f_1\mathbf f_2^{\top}\in\mathbb R^{2N_i\times2N_i}$,
which serves as the conditioning representation supplied to the
inverse-design network.
 
\subsection{Metric and kernel structure induced by the Fourier encoding}
\label{sec::app_metric}
 
We first characterize the geometry of the encoding exactly, rather than only to first order. A direct computation using $ (\sin a-\sin b)^2+(\cos a-\cos b)^2 = 2-2\cos(a-b) $ gives, for any two targets $x$ and $x'$ with $\Delta=x-x'$, 
\begin{equation} \left\|I(x)-I(x')\right\|_2^{2} = \sum_{i=1}^{N_i} 2\bigl(1-\cos(\omega_i\Delta)\bigr) = 4\sum_{i=1}^{N_i} \sin^{2}\!\Bigl(\frac{\omega_i\Delta}{2}\Bigr) \eqqcolon \rho^{2}(\Delta), \label{eq:app_rho} \end{equation} 
and, equivalently, the inner product 
\begin{equation} k_I(\Delta) \coloneqq I(x)^{\top}I(x') = \sum_{i=1}^{N_i}\cos(\omega_i\Delta). \label{eq:app_kernel} \end{equation} 
The kernel is bounded as $|k_I(\Delta)|\le N_i$,
with $k_I(0)=N_i$. Two structural properties follow immediately. 
\emph{(i) Stationarity.} Both the distance $\rho(\Delta)$ and the kernel $k_I(\Delta)$ depend only on the coordinate difference $\Delta$, inducing a translation-invariant geometry in the encoded space. In particular, the local sensitivity 
\begin{equation} \left\| \frac{\partial I(x)}{\partial x} \right\|_2^{2} = \sum_{i=1}^{N_i}\omega_i^{2} \eqqcolon A^{2}, \label{eq:app_amp} \end{equation} 
is independent of $x$: the encoding amplifies all targets uniformly by the constant factor $A$, and for small perturbations $ \rho(\Delta)\approx A|\Delta|. $ 
\emph{(ii) Saturation.} A constant linear amplification alone would not explain the benefit of the encoding, since a fixed rescaling of the input can be absorbed into the first linear layer of any MLP. The essential property of Eq.~\ref{eq:app_rho} is its nonlinear shape: $\rho^{2}(\Delta)$ grows quadratically for $|\Delta|\lesssim\omega_{\max}^{-1}$, then saturates toward the ambient scale $\rho^{2}\le4N_i$, while the kernel $k_I(\Delta)$ correspondingly decorrelates. The encoding therefore converts coordinate differences into a similarity structure with a tunable correlation length set by the frequency ladder $\{\omega_i\}$, a transformation that cannot, in general, be reproduced by a single linear reparameterization of the raw coordinates.
 
\subsection{Relation to spectral bias}
\label{sec::app_ntk}
 
The Fourier-induced kernel structure in Eq.~\ref{eq:app_kernel} provides a kernel-based interpretation of the empirical observations in Sec.~\ref{sec_analy_den}. In the wide-network limit, gradient-descent training of a sufficiently wide MLP can be approximated by its neural tangent kernel (NTK)~\cite{jacot2018ntk}. The components of the target function aligned with the leading kernel eigenfunctions are learned fastest, and MLPs operating on raw low-dimensional coordinates are known to exhibit spectral bias, learning low-frequency dependencies significantly faster than high-frequency ones~\cite{dimensmismat}. Tancik et al.~\cite{tancik2020fourier} showed that composing an MLP with Fourier feature mappings induces an approximately stationary kernel whose bandwidth is controlled by the Fourier spectrum. Consequently, the frequency ladder in Eq.~\ref{eq:app_phi} provides a multi-scale representation of the target coordinates, allowing the network to capture both smooth global trends and sharper target-dependent variations in the design manifold. Although this interpretation does not constitute a formal NTK analysis of the proposed architecture, it provides additional intuition for why Fourier-based conditioning can improve the learnability of target-dependent mappings.
 
\subsection{Injectivity and resolution of the encoding}
\label{sec::app_injective}
 
Because the components of $I$ are periodic, it must be verified that the encoding does not alias distinct targets onto the same representation. 
\begin{proposition}[Injectivity] \label{prop:injective} 
If $\omega_{\min}L<2\pi$, where $L$ is the length of the target domain, then $I$ is injective on that domain. In particular, for the frequency ladder in Eq.~\ref{eq:app_phi} defined on the domain $[0,N_x]$, injectivity holds since $ \omega_{\min}N_x = \frac{1}{N_x} < 2\pi $ for all $N_x\ge1$. 
\end{proposition} 
\begin{proof} 
The pair $ \bigl( \sin(\omega_{\min}x), \cos(\omega_{\min}x) \bigr) $ traces the unit circle with period $ \frac{2\pi}{\omega_{\min}} > L, $ and is therefore injective on the domain. Since this pair already uniquely determines $x$ over the domain, injectivity of one sub-block implies injectivity of the full encoding $I$. 
\end{proof} 

At the opposite end of the frequency ladder, the highest frequency determines the \emph{resolution} of the encoding. For two adjacent targets separated by the minimal grid spacing $\delta$ ($\delta=1$ in index units, corresponding to the half-wavelength spacing of the densest sampling stage in Sec.~\ref{sec::tradetai}), 
\begin{equation} \rho^{2}(\delta) \ge 4\sin^{2}\!\Bigl( \frac{\omega_{\max}\delta}{2} \Bigr) >0. \label{eq:app_resolution} \end{equation} 
Since $ 0<\omega_{\max}\delta<2\pi, $ the lower bound in Eq.~\ref{eq:app_resolution} is strictly positive. Additional contributions from the remaining frequency components further increase the separation between neighboring targets. The frequency ladder therefore simultaneously provides global injectivity and local target discrimination. The lowest frequency component ensures that distinct targets remain identifiable over the entire target domain, while the highest frequencies increase the separation between neighboring targets. Combined with the local amplification property in Eq.~\ref{eq:app_amp}, this guarantees that neighboring targets remain separated by a nonzero distance in the encoded space. Consequently, the encoding preserves both large-scale target structure and fine-scale target variations, making it well suited for dense target-space conditioning.
 
\subsection{Bilinear lifting through the outer product}
\label{sec::app_outer}
 
We now analyze the outer-product construction $\mathbf g=\mathbf f_1\mathbf f_2^{\top}$, whose vectorization satisfies $ \mathrm{vec}(\mathbf g) = \mathbf f_2\otimes\mathbf f_1. $ One clarification is in order. Since $\mathbf g$ is a deterministic smooth function of two scalars, the intrinsic dimension of the conditioning manifold $ \{\mathbf g(\mathbf T)\} $ cannot exceed two, regardless of the lifting. What the outer product increases is the \emph{linear} dimensionality (effective rank) of the family of representations, which is precisely the quantity measured in Fig.~\ref{Figure_witho_dimexp}(f) and the quantity that determines how the downstream network can linearly separate target conditions in its early layers. 
\begin{proposition}[Multiplicative growth of the representation span] \label{prop:rank} 
Let the targets be sampled on a product grid $\mathcal X_1\times\mathcal X_2$, and let $ r_1 = \dim\, \mathrm{span} \{\mathbf f_1(x_1):x_1\in\mathcal X_1\} $ and $ r_2 = \dim\, \mathrm{span} \{\mathbf f_2(x_2):x_2\in\mathcal X_2\}. $ Then 
\begin{equation} 
\dim\, \mathrm{span} \Bigl\{ \mathrm{vec}\, \mathbf g(x_1,x_2) : (x_1,x_2)\in \mathcal X_1\times\mathcal X_2 \Bigr\} = r_1r_2. 
\end{equation} 
By contrast, the concatenated representation $[\mathbf f_1;\mathbf f_2]$ spans a subspace of dimension at most $r_1+r_2$. 
\end{proposition} 
\begin{proof} 
The span of the elementary tensors $ \{ \mathbf f_2(x_2)\otimes\mathbf f_1(x_1) \} $ over a product set equals the tensor product of the individual spans, $ \mathrm{span}\{\mathbf f_2\} \otimes \mathrm{span}\{\mathbf f_1\}, $ whose dimension is $r_1r_2$. The concatenation statement follows from $ [\mathbf f_1;\mathbf f_2] \in \mathrm{span}\{\mathbf f_1\} \oplus \mathrm{span}\{\mathbf f_2\}. $ 
\end{proof} 

The equality requires product-grid sampling and therefore applies to the metalens task. For tasks with functionally dependent target coordinates, such as the Y-splitter benchmark, only the upper bound $r_1r_2$ applies, since the target coordinates do not span a full Cartesian product grid. The lifting therefore expands the linear span of the representation family multiplicatively, whereas concatenation increases the span only additively. This provides a concrete representational advantage of the outer product beyond the Fourier encoding itself. Although the intrinsic dimension of the conditioning manifold remains two, the larger linear span can allow downstream layers to separate target
conditions more effectively, which is consistent with the increased effective rank observed experimentally in Fig.~\ref{Figure_witho_dimexp}(f). The outer product also couples the two coordinates in the induced metric. Using $ \|\mathbf a\mathbf b^{\top}\|_F = \|\mathbf a\|_2 \|\mathbf b\|_2 $ and the trace identity $ \mathrm{tr} \Bigl( (\mathbf f_1'\mathbf f_2'^{\top})^{\top} \mathbf f_1\mathbf f_2^{\top} \Bigr) = (\mathbf f_1^{\top}\mathbf f_1') (\mathbf f_2^{\top}\mathbf f_2'), $ we obtain 
\begin{equation} 
\bigl\| \mathbf g-\mathbf g' \bigr\|_F^2 = \|\mathbf f_1\|_2^2 \|\mathbf f_2\|_2^2 + \|\mathbf f_1'\|_2^2 \|\mathbf f_2'\|_2^2 - 2 (\mathbf f_1^{\top}\mathbf f_1') (\mathbf f_2^{\top}\mathbf f_2'). 
\label{eq:app_sep_general} 
\end{equation} 
For unit-normalized features, Eq.~\ref{eq:app_sep_general} reduces to 
\begin{equation} 
\bigl\| \mathbf g-\mathbf g' \bigr\|_F^2 = 2\bigl(1-c_1c_2\bigr), \qquad c_j \coloneqq \mathbf f_j^{\top}\mathbf f_j' \in[-1,1]. 
\label{eq:app_sep} 
\end{equation} 
For intuition, Eq.~\ref{eq:app_sep}
shows that, after normalization,
coincidence of the lifted representations
requires simultaneous agreement in both coordinate embeddings. More generally, Eq.~\ref{eq:app_sep_general} demonstrates that the geometry of the lifted representation depends multiplicatively on the similarities of the two coordinate embeddings. Moreover, the derivatives 
\begin{equation} 
\frac{\partial\mathbf g}{\partial x_1} = \Bigl( \frac{\partial\mathbf f_1}{\partial x_1} \Bigr) \mathbf f_2^{\top}, \qquad \frac{\partial\mathbf g}{\partial x_2} = \mathbf f_1 \Bigl( \frac{\partial\mathbf f_2}{\partial x_2} \Bigr)^{\!\top}, 
\end{equation} 
show that a perturbation of either coordinate is distributed across every entry of the conditioning matrix. Consequently, target variations are not confined to a small subset of input channels of the downstream network, but instead influence the conditioning representation globally.

\subsection{Gram-matrix structure and mode collapse}
\label{sec::app_collapse}

The preceding results combine into a quantitative account of the mode-collapse behavior observed in Fig.~\ref{Figure_witho_dimexp}. 

\paragraph{Gram-matrix structure and spectral properties.} Consider the Gram matrix of the conditioning representations over the training targets $\{\mathbf T_1,\ldots,\mathbf T_M\}$, 
\begin{equation} 
\bigl(K_{\mathbf g}\bigr)_{mn} = \Bigl\langle \mathrm{vec}\,\mathbf g(\mathbf T_m), \mathrm{vec}\,\mathbf g(\mathbf T_n) \Bigr\rangle = \bigl(K_1\bigr)_{mn} \bigl(K_2\bigr)_{mn}, 
\label{eq:app_gram} 
\end{equation} 
where $K_1$ and $K_2$ denote the Gram matrices induced by the $x$- and $y$-coordinate embeddings, respectively, so that the outer-product conditioning construction induces a Hadamard-factorized Gram matrix. Equivalently, 
\begin{equation} 
K_{\mathbf g} = K_1 \circ K_2, 
\end{equation} 
where $\circ$ denotes the Hadamard (elementwise) product. By the Schur product theorem, $K_{\mathbf g}$ remains positive semidefinite whenever $K_1$ and $K_2$ are positive semidefinite. Moreover, since $ K_1 \succeq \lambda_{\min}(K_1) I, $ the monotonicity of the Hadamard product with respect to the Loewner order implies $ K_1 \circ K_2 \succeq \lambda_{\min}(K_1)\,(I\circ K_2) = \lambda_{\min}(K_1)\, \mathrm{diag}\!\bigl((K_2)_{11},\ldots,(K_2)_{MM}\bigr). $ Therefore, 
\begin{equation} 
\lambda_{\min}(K_{\mathbf g}) \ge \lambda_{\min}(K_1)\, \min_n (K_2)_{nn}. 
\label{eq:app_schur} 
\end{equation} 
By symmetry, 
\begin{equation} 
\lambda_{\min}(K_{\mathbf g}) \ge \lambda_{\min}(K_2)\, \min_n (K_1)_{nn}. 
\end{equation} 
Hence, 
\begin{equation} 
\lambda_{\min}(K_{\mathbf g}) \ge \max\!\Bigl( \lambda_{\min}(K_1)\min_n (K_2)_{nn}, \; \lambda_{\min}(K_2)\min_n (K_1)_{nn} \Bigr). 
\label{eq:app_schur_sym} 
\end{equation} 
Since $(K_1)_{nn} = \|\mathbf f_1(\mathbf T_n)\|^2$, $(K_2)_{nn} = \|\mathbf f_2(\mathbf T_n)\|^2$, the lower bound remains positive provided that the feature norms do not collapse to zero. Thus, provided that the feature norms remain bounded away from zero, the lifted representation preserves a positive lower bound on its smallest eigenvalue while remaining positive semidefinite~\cite{roger1994topics}. 

\begin{remark} 
On a full product grid, $K_1$ and $K_2$ are necessarily singular because coordinate values repeat across targets, so the lower bound in Eq.~\ref{eq:app_schur_sym} may become vacuous. The bound is more informative when applied to subsets of targets with pairwise-distinct coordinates, for which it guarantees that the outer-product lifting preserves the nondegeneracy of the per-coordinate Gram matrices. In particular, Eq.~\ref{eq:app_schur_sym} provides only a one-sided guarantee. It certifies nondegeneracy of $K_{\mathbf g}$ when sufficient nondegeneracy is already present in the per-coordinate Gram matrices, but does not fully characterize the rank of $K_{\mathbf g}$. Indeed, the outer-product lifting can substantially increase the span of the representation family, with \[ \mathrm{rank}(K_{\mathbf g}) \le \mathrm{rank}(K_1)\, \mathrm{rank}(K_2), \] and Proposition~\ref{prop:rank} shows that this upper bound can be attained on a product grid. 
\end{remark} 

\paragraph{Correlation suppression and mode-collapse mitigation.}

The off-diagonal entries of the lifted Gram matrix factor multiplicatively. For two distinct targets $m\neq n$,
\begin{equation}
(K_{\mathbf g})_{mn} = (K_1)_{mn}(K_2)_{mn}.
\label{equa_app_corr}
\end{equation}
The similarity between two target conditions is therefore gated jointly by the two coordinate embeddings: it depends jointly on the similarities in both coordinates, and it vanishes whenever either embedding renders the pair orthogonal. This contrasts with additive constructions such as concatenation, where similarity contributions from different coordinates are accumulated rather than coupled multiplicatively. Equation~\ref{equa_app_corr} thus shows that the outer-product lifting makes the conditioning representation more selective in distinguishing target pairs, complementing the rank growth of Proposition~\ref{prop:rank}.

This selectivity is most relevant in the regime where mode collapse arises. When the network is conditioned directly on the raw coordinates, nearby targets produce nearly identical representations; the corresponding rows of the Gram matrix become similar, so the per-target gradient contributions align and training is effectively driven by a shared objective across neighboring targets, encouraging a common structure and reducing output diversity. Dimension expansion counteracts this in two complementary ways: the Fourier encoding separates nearby targets through the stationary multi-scale kernel of Sec.~\ref{sec::app_metric}, and the outer-product lifting requires agreement in both coordinates for two conditioning matrices to coincide, as reflected in Eq.~\ref{eq:app_sep_general}. We emphasize that these arguments identify mechanisms rather than guarantees; together they offer a plausible account of the larger amplification factors, higher effective ranks, and reduced mode-collapse behavior observed in Fig.~\ref{Figure_witho_dimexp}(e,f).

\end{document}

%% file: math_commands.tex
\usepackage{amsmath,amsfonts,bm}

\def\eqref#1{equation~\ref{#1}}

\def\1{\bm{1}}

\DeclareMathAlphabet{\mathsfit}{\encodingdefault}{\sfdefault}{m}{sl}
\SetMathAlphabet{\mathsfit}{bold}{\encodingdefault}{\sfdefault}{bx}{n}

